\documentclass[12pt,a4paper,oneside]{article}

\title{Notes on Computational Graph and \\ Jacobian Accumulation}
\author{Yichong Zhou \\ yichong.zhou@gmail.com}
\date{Dec 24, 2020}

\usepackage{amssymb}
\usepackage{amsthm}
\usepackage[fleqn]{amsmath}
\usepackage{mathtools}

\PassOptionsToPackage{hyphens}{url}\usepackage{hyperref}

\usepackage{booktabs}
\usepackage{graphicx,float}
\usepackage[top=1in, bottom=1.25in, left=0.85in, right=0.85in]{geometry}

\usepackage{xcolor,soul}
\usepackage[utf8]{inputenc}

\usepackage[depth=3]{bookmark}
\usepackage{tikz}
\usepackage{pgf}
\usetikzlibrary{positioning}
\usetikzlibrary{arrows,automata} 
\usepackage{caption}
\usepackage[capitalise]{cleveref}
\usepackage{subcaption}
\usepackage{titlesec}

\theoremstyle{definition}
\newtheorem{definition}{Definition}[section]
\newtheorem{lemma}{Lemma}
\newtheorem{proposition}{Proposition}

\titleformat{\section}
	{\normalfont\normalsize\bfseries}{\thesection.}{.6em}{}

\numberwithin{figure}{section} 
\numberwithin{equation}{section}
\begin{document}
	
\maketitle

\begin{abstract}
The optimal calculation order of a computational graph 
can be represented by a set of algebraic expressions. 
Computational graph and algebraic expression 
both have close relations and significant differences, 
this paper looks into these relations and differences, 
making plain their interconvertibility. 
By revealing different types of multiplication relations in algebraic expressions 
and their elimination dependencies in line-graph, 
we establish a theoretical limit on the efficiency of face elimination. 
\end{abstract}

\section{Face Elimination and D* Algorithm}

After a sparse computational graph or subgraph is linearized, 
i.e., all the derivatives denoted by edges
and vertices in the computational graph are evaluated at a point,
there are many elimination techniques to accumulate derivative values,  
including vertex elimination, edge elimination and face elimination.

Naumann states in his paper that all the elimination techniques mentioned above
``are based on the elimination of transitive dependences between variables in $F$. 
A variable $v_j$ depends transitively on $v_i$ via $v_k$ if $i\prec k\prec j$.
In general, there is no structural representation for eliminating such dependences 
in $G$; that is, it cannot be expressed by modifying either $V$ or $E$. 
A richer data structure is required, namely, a directed variant of the line graph of $G$.''
``All intermediate vertices belong to exactly two directed complete bipartite subgraphs
of $\tilde{G}$. They are minimal in one and maximal in the other.
Intermediate vertices in $G$ are mapped onto complete bipartite subgraphs (or bicliques)
$K_{\nu,\mu}$ of $\tilde{G}$'',
He continues to explain,
``The elimination of transitive dependences can be interpreted 
as the elimination of edges in $\tilde{G}$. 
The modification of the dual c-graph is referred to as face elimination 
in order to distinguish between edge elimination in $\tilde{G}$ and $G$.''
(\cite{naumann2003} 2003, p5-8)

Griewank and Walther state in their book that in edge elimination ``we are forced to 
simultaneously merge them with all their predecessor or successor edges, respectively. 
Sometimes we may wish to do that only with one of them. 
However, that desire cannot be realized by a simple modification
of the linearized computational graph. 
Instead we have to consider the so-called line-graph of 
the computational graph with a slight extension.'' 
They continue:
``Geometrically, we may interpret the edges of the line-graph as faces of the
extended computational graph. Therefore, we will refer to them as faces.'' 
(\cite{griewank2008} 2008, p204)

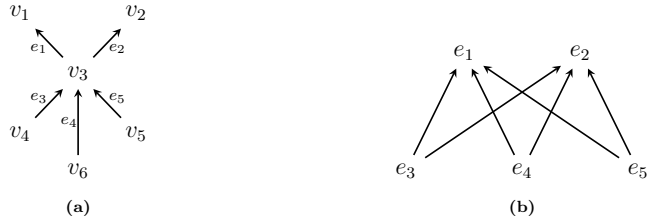
\begin{figure}[H] 
	\centering
	\resizebox{.7\textwidth}{!}{
		\begin{subfigure}[b]{0.4\textwidth}
			\centering
			\begin{tikzpicture}[>=stealth,thick]
			\usetikzlibrary{calc}
			\tikzstyle{state}=[]
			\node[state] (q2) [xshift=-1cm] {$v_1$};
            \node[state] (q3) [xshift=1cm] {$v_2$};
			\node[state] (q4) [below  =0.5cm of q3, xshift=-1cm] {$v_3$};
			\node[state] (q5) [below  =0.5cm of q4, xshift=-1cm] {$v_4$};
			\node[state] (q6) [below  =0.5cm of q4, xshift=1cm] {$v_5$};
			\node[state] (q7) [below  =0.1cm of q6, xshift=-1cm] {$v_6$};
			
			\path[<-] 
			
			(q2) edge node[below=-4pt, xshift=-5pt, font=\scriptsize]{$e_1$} (q4)
			(q3) edge node[below=-4pt, xshift=5pt, font=\scriptsize]{$e_2$} (q4)
			(q4) edge node[above=-4pt, xshift=-5pt, font=\scriptsize]{$e_3$} (q5)
			edge node[above=-4pt, xshift=5pt, font=\scriptsize]{$e_5$} (q6)
			edge node[left=-4pt,font=\scriptsize]{$e_4$} (q7)
			;
			\end{tikzpicture}
			\caption{}
		\end{subfigure}
		\begin{subfigure}[b]{0.5\textwidth}
			\centering
			\begin{tikzpicture}[>=stealth,thick]
			\usetikzlibrary{calc}
			\tikzstyle{state}=[]
			
			\node[state] (q1) [xshift=-1cm] {$e_1$};
			\node[state] (q2) [xshift=1cm] {$e_2$};
			\node[state] (q3) [below  =1.5cm of q1, xshift=-1cm] {$e_3$};
			\node[state] (q4) [below  =1.5cm of q2, xshift=-1cm] {$e_4$};
			\node[state] (q5) [below  =1.5cm of q2, xshift=1cm] {$e_5$};
			
			\path[<-] 
			(q1) edge (q3)
			edge (q4)
			edge (q5)
			(q2) edge (q3)
			edge (q4)
			edge (q5)
			;
			\end{tikzpicture}
			\caption{}
		\end{subfigure}
	}
	\caption{Line-Graph and Biclique}
	\label{fig:biclique}
\end{figure}
The transitive dependences Naumann talks about are actually 
the adjacent edge paths in computational graph, 
i.e., one in-going edge and one out-going edge of a vertex form a in-out vertex path. 
e.g., $(v_4,v_3,v_1)$ in \cref*{fig:biclique} (a) is an adjacent edge path.
As multiple paths could overlap on the same edge and vertex, 
elimination of one edge could result in elimination of multiple paths, 
e.g., adjacent edge path $(v_4,v_3,v_1)$ and $(v_4,v_3,v_2)$ overlap on edge $e_3$,
elimination of $e_3$ could result in elimination of $(v_4,v_3,v_1)$ and $(v_4,v_3,v_2)$. 
It seems impossible to eliminate transitive dependences 
or adjacent edge paths one by one in computational graph. 
However, adjacent edges connected by one vertex in $G$ form 
a complete bipartite subgraph (or biclique) in $\tilde{G}$. 
Since every adjacent edge path in $G$ has an edge in $\tilde{G}$ representing it, e.g., 
\cref*{fig:biclique} (b) is the line-graph of computational graph \cref*{fig:biclique} (a),
edge $(e_3,e_1)$ in \cref*{fig:biclique} (b) represents 
the adjacent edge path $(v_4,v_3,v_1)$ in \cref*{fig:biclique} (a),   
it can be eliminated in the corresponding line-graph, i.e., 
face elimination in $\tilde{G}$ is equivalent to adjacent edge path elimination in $G$. 

In an article uploaded in 2017, Naumann and Utke make a further explanation: 
``AD exploits the associativity of the chain rule to compute derivatives on the basis of 
the dependency information structurally represented by paths in the computational graph.
The derivative computation can be represented by elimination
steps in the computational graph.''
``Face elimination reduces the granularity of an elimination step 
to a single multiplication operation.''
``We define the optimal Jacobian accumulation (OJA) under chain rule arithmetic 
as a computation of Bauer's formula with minimal multiplication count.''
``Until now there has been no proof that one can attain the minimum with face elimination either.
Face elimination does not cover all possibilities of computing Bauer's formula, 
particularly the commutativity and distributivity of the multiplications.
We outline a proof showing that there is no exploit of chain rule arithmetic manipulations 
in Bauer's formula that can undercut an optimal face elimination, 
thereby establishing that face elimination can indeed solve OJA.''
(\cite{naumann-utke} p1-2)
\\ 

\noindent
As for the rule of face elimination, Naumann described as follows: \\
For all intermediate faces $(i,j) \in \tilde{E}_Z$, \\ ``
(1). If there exists a vertex $k \in \tilde{V}$ such that $P_k=P_i$ and $S_k=S_j$, 
then set $c_k=c_k+c_jc_i$ (absorption); 
else $\tilde{V}=\tilde{V}\cup \{k'\}$ with a new vertex $k'$ 
such that $P_{k'}=P_i$ and $S_{k'}=S_j$ (fillin) 
and labeled with $c_{k'}=c_jc_i$. \\ 
(2). Remove edge $(i,j)$ from $\tilde{E}$. \\ 
(3). Remove $i\in \tilde{V}$ if it is isolated. 
Otherwise, if there exists a vertex $i' \in \tilde{V}$ 
such that $P_{i'}=P_i$ and $S_{i'}=S_i$, 
then set $c_i=c_i+c_{i'}$ (merge); remove $i'$. \\
(4). Repeat Step (3) for $j\in \tilde{V}$. '' \\
(\cite{naumann2003} 2003, p8) \\

In step (1) of Naumann's rule, the absorption $c_k=c_k+c_jc_i$ means that 
we can perform this operation by modifying and reusing $k$ 
instead of creating a new vertex. 
On the principle of reusing vertices, 
when $|S_{i}|=1$, i.e. $S_{i}=\{j\}$, 
we can perform fillin operation by $c_i=c_jc_i$ and $S_{i}=S_j$;
when $|P_{j}|=1$, i.e. $P_{j}=\{i\}$, 
we can perform fillin operation by $c_j=c_jc_i$ and $P_{j}=P_i$; 

Griewank and Walther also presented their version of modification rules, 
they explained that these modifications will not alter the accumulated values $a_{ij}$,
``but hopefully simplifying their computation''.
In addition to the merge and remove operations similar to those in Naumann's description
they added the ``interior vertex split'' operation:
``make a copy of any interior vertex by assigning the same value and
the same predecessor or successor set while splitting the other set,
i.e., reattaching some of the incoming or some of the outgoing edges
to the new copy.'' 
They continued to explain:
``we can split and merge until every interior vertex $e$ is simply connected in that
it has exactly one incoming and one outgoing edge connecting it to the vertices
$o_j$ and $d_i$ for which then $c_e = a_{ij}$.''
``Applying the modifications described above, one may generate new directed
graphs that are no longer line-graphs of an underlying computational graph.''
(\cite{griewank2008} p205)

According to the splitting Griewank and Walther mentioned, 
in addition to those presented by Naumann, 
there are more operations that can be performed in line-graph: 
If $P_k=P_i$ and $S_k\subset S_j$, 
we can perform absorption by $c_k=c_k+c_jc_i$, $S_j=S_j-S_k$ without removing $(i,j)$. 
If $P_k\subset P_i$ and $S_k=S_j$, 
we can perform absorption by $c_k=c_k+c_jc_i$, $P_i=P_i-P_k$ without removing $(i,j)$. 
If $P_k=P_i$ and $S_k\supset S_j$, 
we can perform fillin operation 
by $\tilde{V}=\tilde{V}\cup \{k'\}$, $c_{k'}=c_jc_i+c_k$, $S_k=S_k-S_j$ 
if $|S_i|>1\cap|P_j|>1$; 
$c_i=c_jc_i+c_k$, $S_k=S_k-S_j$, $S_{i}=S_j$ 
if $|S_i|=1\cap|P_j|\ge1$; 
$c_j=c_jc_i+c_k$, $S_k=S_k-S_j$, $P_{j}=P_i$  
if $|S_i|\ge1\cap|P_j|=1$. 
If $P_k\supset P_i$ and $S_k=S_j$, 
we can perform fillin operation 
by $\tilde{V}=\tilde{V}\cup \{k'\}$, $c_{k'}=c_jc_i+c_k$, $P_k=P_k-P_i$ 
if $|S_i|>1\cap|P_j|>1$; 
$c_i=c_jc_i+c_k$, $P_k=P_k-P_i$, $S_{i}=S_j$ 
if $|S_i|=1\cap|P_j|\ge1$; 
$c_j=c_jc_i+c_k$, $P_k=P_k-P_i$, $P_{j}=P_i$ 
if $|S_i|\ge1\cap|P_j|=1$. 
If there exists a vertex $i' \in \tilde{V}$ 
such that $P_{i'}\supseteq P_i$ and $S_{i'}=S_i$, 
then we can perform merge operation by $c_i=c_i+c_{i'}$, $P_{i'}=P_{i'}-P_i$;  
or $P_{i'}=P_i$ and $S_{i'}\supseteq S_i$, 
then we can perform merge operation by $c_i=c_i+c_{i'}$, $S_{i'}=S_{i'}-S_i$.  

In an article published in 2002, 
Griewank and Naumann explain the disadvantages of face elimination:  
``it is clear that face elimination offers an even larger number of choices than edge elimination'' 
``a poor choice of edge or face elimination sequences might lead to an exponential growth of the intermediate graph structure and thus the overall cost. In other words the extra freedom may lead one astray, unless suitable safeguards can be found.''
``If implemented dynamically, it also requires examining and possibly reconnecting all predecessors of the faces origin and all successors of its destination. Concerning the ratio between the number of floating point operations and the number of memory accesses, dynamic face elimination is thus a true nightmare.''
``we may view edge and vertex eliminations as chunkings of face eliminations to achieve a better ratio between computation and communication operations.''
``in any case we must therefore strive for a static implementation, where all decisions regarding the elimination sequence and the necessary memory allocations as well as some of the address and pointer calculation are performed at compile time. They may then be hardwired into a code for evaluating the Jacobian $F'(x)$ efficiently at many different arguments $x$''.
(\cite{griewank2002naumann} 2002, p8-9)

Brian Guenter states in his paper that 
``D* computes efficient symbolic derivatives 
by symbolically executing the expression graph at compile time 
to eliminate common subexpressions and 
by exploiting the special nature of the graph that represents the derivative of a function.''
``factorable subgraphs are defined by a dominator or postdominator node at a branch in the graph. 
If a dominator node $b$ has more than one child, 
or if a post-dominator node b has more than one parent, then $b$ is a factor node. 
If $c$ is dominated by a factor node $b$ and has more than one parent, 
or $c$ is postdominated by $b$ and has more than one child, then $c$ is a factor base of $b$. 
A factor subgraph, $[b,c]$ consists of a factor node $b$, 
a factor base $c$ of $b$, and those nodes on any path from $b$ to $c$.''
He then describes a factoring algorithm to eliminate factor subgraph in the derivative graph 
and convert it into subgraph edge, adding that 
``factorization does not change the value of the sum of products expression
which represents the derivative so factor subgraphs can be factored in any order''. 
He says 
``The solution to the problem of common factor subgraphs is to 
count the number of times each factor subgraph $[i,j]$ appears in the $nm$ derivative graphs. 
The factor subgraph which appears most often is factored first. 
If factor subgraph $[k,l]$ disappears in some derivative graphs 
as a result of factorization then the count of $[k,l]$ is decremented. 
To determine if factorization has eliminated $[k,l]$ from some derivative graph $f_{ij}$ 
it is only necessary to count the children of a dominator node or the parents of a postdominator node.
If either is one the factor subgraph no longer exists. 
The counts of the $[k,l]$ are efficiently updated during factorization 
by observing if either node of a deleted edge is either a factor or factor base node.
Ranking of the $[k,l]$ can be done efficiently with a priority queue.''
(\cite{Guenter-2007a} 2007, p1-7)
\\

\section{Differentiation Graph}

We realize that, in $\frac{\mathrm{d} y}{\mathrm{d} x}$, 
the relation between $\mathrm{d} y$ and $\mathrm{d} x$ is a partial order 
$\mathrm{d}y \geq \mathrm{d}x$.
e.g., $y=f(g(x))$, 
if we write $\frac{\mathrm d f}{\mathrm d g}$ and $\frac{\mathrm d g}{\mathrm d x}$ 
as $\mathrm d f \to \mathrm d g$ and $\mathrm d g \to \mathrm d x$ respectively, 
the transitive relation of partial order 
between $\mathrm{d} f$ and $\mathrm{d} x$ will be denoted by chain rule. 

$$
\frac{\mathrm d f}{\mathrm d x}
=
\,
\frac{\mathrm d f}{\mathrm d g} \frac{\mathrm d g}{\mathrm d x}
\implies
\,
\textrm d f \to \textrm d x = \mathrm d f \to \mathrm d g \to \mathrm d x
$$
\\
$\textrm d f \to \textrm d x$ means $\textrm d f$ depends on $\textrm d x$  
or there is a path of dependence from $\textrm d f$ to $\textrm d x$.
\\
The relation between adjacent arrows in the path is multiplication.
To get the value of $\textrm d f \to \textrm d x$, 
we multiply all the derivatives of adjacent pairs along the path 
from $\mathrm d f$ to $\mathrm d x$, i.e.,

$$
\mathrm{d} f \to \mathrm{d} g_1 \to \mathrm{d} g_2 \to \cdots \to \mathrm{d} g_n \to \mathrm{d} x
\implies
\frac{\mathrm d f}{\mathrm d g_1} \frac{\mathrm d g_1}{\mathrm d g_2}
\cdots \frac{\mathrm d g_{n}}{\mathrm d x}
$$
\\
As for the multivariate chain rule, e.g., $y=f(g_1(x), g_2(x), g_3(x))$. 
The derivative $\frac{\mathrm d f}{\mathrm d x}$ would be like this: 

$$
\frac{\mathrm d f}{\mathrm d x}
=
\frac{\mathrm d f}{\mathrm d g_1} \frac{\mathrm d g_1}{\mathrm d x} +
\frac{\mathrm d f}{\mathrm d g_2} \frac{\mathrm d g_2}{\mathrm d x} +
\frac{\mathrm d f}{\mathrm d g_3} \frac{\mathrm d g_2}{\mathrm d x}
\implies 
\frac{\mathrm d f}{\mathrm d x} 
\,
=
\,
\mathrm d f
\begin{cases}
\to \mathrm d g_1  \to \mathrm d x        \\
\to \mathrm d g_2  \to \mathrm d x        \\
\to \mathrm d g_3  \to \mathrm d x
\end{cases}
$$
\\
Notice that $\mathrm d f$ is divided into $\mathrm d g_1$, $\mathrm d g_2$ and $\mathrm d g_3$ 
three different paths and we sum up all the products of different paths in the end. 
The relation between different paths is addition, i.e.,

$$
\mathrm d f
\begin{cases}
\to \mathrm d g_1  \to \mathrm d x        \\
\to \mathrm d g_2  \to \mathrm d x        \\
\;\vdots                                    \\
\to \mathrm d g_n  \to \mathrm d x
\end{cases}
\implies
\frac{\mathrm d f}{\mathrm d g_1} \frac{\mathrm d g_1}{\mathrm d x} +
\frac{\mathrm d f}{\mathrm d g_2} \frac{\mathrm d g_2}{\mathrm d x} +
\cdots +
\frac{\mathrm d f}{\mathrm d g_n} \frac{\mathrm d g_n}{\mathrm d x}
$$
\\
Suppose there are intermediate variables $\left[g_1,g_2,g_3\right]$ and $h$ between $y$ and $x$, e.g., \\ $y=f(g_1(h(x)), g_2(h(x)), g_3(h(x)))$, then

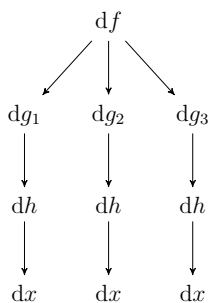
\begin{figure}[H] 
\centering
\resizebox{.4\textwidth}{!}{
\begin{tikzpicture}[shorten >=1pt,node distance=2.5cm,>=stealth']
\usetikzlibrary{calc}

\tikzstyle{state}=[]

\node[state] (q0) {$\mathrm{d}f$};
\node[state] (q-1) [above  =0.3cm of q0]
{$\big\Downarrow$};
\node[state] (q-2) [above  =0.3cm of q-1]
{\large$
\frac{\mathrm d f}{\mathrm d x} 
\,
=
\,
\mathrm d f
\begin{cases}
\to \mathrm d g_1  \to \mathrm d h  \to \mathrm d x        \\
\to \mathrm d g_2  \to \mathrm d h  \to \mathrm d x        \\
\to \mathrm d g_3  \to \mathrm d h  \to \mathrm d x
\end{cases}
$};
\node[state] (q-3) [above  =0.3cm of q-2]
{$\big\Downarrow$};
\node[state] (q-4) [above  =0.3cm of q-3]
{\Large$
\frac{\mathrm d f}{\mathrm d x}
=
\frac{\mathrm d f}{\mathrm d g_1} \frac{\mathrm d g_1}{\mathrm d h} \frac{\mathrm d h}{\mathrm d x} +
\frac{\mathrm d f}{\mathrm d g_2} \frac{\mathrm d g_2}{\mathrm d h} \frac{\mathrm d h}{\mathrm d x} +
\frac{\mathrm d f}{\mathrm d g_3} \frac{\mathrm d g_3}{\mathrm d h} \frac{\mathrm d h}{\mathrm d x}
$};

\node[state] (q1) [below  =1cm of q0, xshift=-1.5cm] {$\mathrm{d}g_1$};
\node[state] (q2) [below  =1cm of q0] {$\mathrm{d}g_2$};
\node[state] (q3) [below  =1cm of q0, xshift=1.5cm] {$\mathrm{d}g_3$};
\node[state] (q4) [below  =1cm of q1] {$\mathrm{d}h$};
\node[state] (q5) [below  =1cm of q2] {$\mathrm{d}h$};
\node[state] (q6) [below  =1cm of q3] {$\mathrm{d}h$};
\node[state] (q7) [below  =1cm of q4] {$\mathrm{d}x$};
\node[state] (q8) [below  =1cm of q5] {$\mathrm{d}x$};
\node[state] (q9) [below  =1cm of q6] {$\mathrm{d}x$};

\path[->] 
(q0) edge (q1)
edge (q2)
edge (q3)
(q1) edge (q4)
(q2) edge (q5)
(q3) edge (q6)
(q4) edge (q7)
(q5) edge (q8)
(q6) edge (q9)

;

\end{tikzpicture}
}
\caption{The transformation of differentiation expression}\label{fig:diff-expr-transform}
\end{figure}

We notice that the above structure forms a rooted tree hierarchy.

In practice, instead of seeing $h$ in different branches 
as different subexpressions with the same structure, 
computing $\frac{\mathrm d h}{\mathrm d x}$ and multiplying them by the rest derivatives in each branch respectively, 
we would like to factor out the common factor, compute and multiply by it once for all:

$$
\frac{\mathrm d f}{\mathrm d g_1} \frac{\mathrm d g_1}{\mathrm d h} \frac{\mathrm d h}{\mathrm d x} +
\frac{\mathrm d f}{\mathrm d g_2} \frac{\mathrm d g_2}{\mathrm d h} \frac{\mathrm d h}{\mathrm d x} +
\frac{\mathrm d f}{\mathrm d g_3} \frac{\mathrm d g_3}{\mathrm d h} \frac{\mathrm d h}{\mathrm d x}
\implies
\left(
\frac{\mathrm d f}{\mathrm d g_1} \frac{\mathrm d g_1}{\mathrm d h} +
\frac{\mathrm d f}{\mathrm d g_2} \frac{\mathrm d g_2}{\mathrm d h} +
\frac{\mathrm d f}{\mathrm d g_3} \frac{\mathrm d g_3}{\mathrm d h}
\right)
\frac{\mathrm d h}{\mathrm d x}
$$
\\

Similar to the above factorization, 
when we factor out the common factor $\frac{\mathrm d h}{\mathrm d x}$, 
we merge the $\mathrm d h$ in the graph.

\begin{figure}[H] 
	\centering
	\resizebox{.40\textwidth}{!}{
	\begin{tikzpicture}[>=stealth,thick]
\usetikzlibrary{calc}

\tikzstyle{state}=[]

\node[state] (q0) {$\mathrm{d}f$};
\node[state] (q1) [below  =1cm of q0, xshift=-1.5cm] {$\mathrm{d}g_1$};
\node[state] (q2) [below  =1cm of q0] {$\mathrm{d}g_2$};
\node[state] (q3) [below  =1cm of q0, xshift=1.5cm] {$\mathrm{d}g_3$};
\node[state] (q4) [below  =1cm of q1] {$\mathrm{d}h$};
\node[state] (q5) [below  =1cm of q2] {$\mathrm{d}h$};
\node[state] (q6) [below  =1cm of q3] {$\mathrm{d}h$};
\node[state] (q7) [below  =1cm of q4] {$\mathrm{d}x$};
\node[state] (q8) [below  =1cm of q5] {$\mathrm{d}x$};
\node[state] (q9) [below  =1cm of q6] {$\mathrm{d}x$};

\node[state] (arr1) [xshift=2.8cm, yshift=-1.5cm]{$\xRightarrow{\;\frac{\mathrm{d}h}{\mathrm{d}x}\,\;}$};

\node[state] (q20) [xshift=5.5cm] {$\mathrm{d}f$};
\node[state] (q21) [below  =1cm of q20, xshift=-1.5cm] {$\mathrm{d}g_1$};
\node[state] (q22) [below  =1cm of q20] {$\mathrm{d}g_2$};
\node[state] (q23) [below  =1cm of q20, xshift=1.5cm] {$\mathrm{d}g_3$};
\node[state] (q24) [below  =1cm of q22] {$\mathrm{d}h$};
\node[state] (q25) [below  =1cm of q24] {$\mathrm{d}x$};
\path[->] 
(q0) edge (q1)
edge (q2)
edge (q3)
(q1) edge (q4)
(q2) edge (q5)
(q3) edge (q6)
(q4) edge (q7)
(q5) edge (q8)
(q6) edge (q9)

(q20) edge (q21)
edge (q22)
edge (q23)
(q21) edge (q24)
(q22) edge (q24)
(q23) edge (q24)
(q24) edge (q25)
;

\end{tikzpicture}
	}
\caption{}\label{fig:graph-factoring}
\end{figure}

Note that we not only merged all the $\mathrm d h$ nodes, 
but also all the $\mathrm d x$ nodes and the edges between $\mathrm d h$ and $\mathrm d x$.

The same rule also applies to all the leaf nodes of same variable 
because $\frac{\mathrm d x}{\mathrm d x}=1$.

\begin{figure}[H] 
	\centering
	\resizebox{.7\textwidth}{!}{
	\begin{tikzpicture}[>=stealth,thick]
\usetikzlibrary{calc}

\tikzstyle{state}=[]

\node[state] (eq) [above  =1.5cm of q20] 
{$
\frac{\mathrm d f}{\mathrm d g_1} \frac{\mathrm d g_1}{\mathrm d x} +
\frac{\mathrm d f}{\mathrm d g_2} \frac{\mathrm d g_2}{\mathrm d x} +
\frac{\mathrm d f}{\mathrm d g_3} \frac{\mathrm d g_2}{\mathrm d x}
\implies 
\left(
\frac{\mathrm d f}{\mathrm d g_1} \frac{\mathrm d g_1}{\mathrm d x} +
\frac{\mathrm d f}{\mathrm d g_2} \frac{\mathrm d g_2}{\mathrm d x} +
\frac{\mathrm d f}{\mathrm d g_3} \frac{\mathrm d g_2}{\mathrm d x}
\right)
\frac{\mathrm d x}{\mathrm d x}
\implies 
\left(
\frac{\mathrm d f}{\mathrm d g_1} \frac{\mathrm d g_1}{\mathrm d x} +
\frac{\mathrm d f}{\mathrm d g_2} \frac{\mathrm d g_2}{\mathrm d x} +
\frac{\mathrm d f}{\mathrm d g_3} \frac{\mathrm d g_2}{\mathrm d x}
\right)
$};

\node[state] (q0) {$\mathrm{d}f$};
\node[state] (q1) [below  =1cm of q0, xshift=-1.5cm] {$\mathrm{d}g_1$};
\node[state] (q2) [below  =1cm of q0] {$\mathrm{d}g_2$};
\node[state] (q3) [below  =1cm of q0, xshift=1.5cm] {$\mathrm{d}g_3$};
\node[state] (q4) [below  =1cm of q1] {$\mathrm{d}x$};
\node[state] (q5) [below  =1cm of q2] {$\mathrm{d}x$};
\node[state] (q6) [below  =1cm of q3] {$\mathrm{d}x$};

\node[state] (arr1) [xshift=2.8cm, yshift=-1.8cm]{$\xRightarrow[\frac{\mathrm{d}x}{\mathrm{d}x}\;]{\;\,(\;)\,\,\;}$};

\node[state] (q10) [xshift=5.5cm]{$\mathrm{d}f$};
\node[state] (q11) [below  =1cm of q10, xshift=-1.5cm] {$\mathrm{d}g_1$};
\node[state] (q12) [below  =1cm of q10] {$\mathrm{d}g_2$};
\node[state] (q13) [below  =1cm of q10, xshift=1.5cm] {$\mathrm{d}g_3$};
\node[state] (q14) [below  =1cm of q12] {$\mathrm{d}x$};

\node[state] (arr2) [xshift=8.3cm, yshift=-1.7cm]{$\xRightarrow[\,]{\;\,(\;)\,\,\;}$};

\node[state] (q20) [xshift=11cm]{$\mathrm{d}f$};
\node[state] (q21) [below  =1cm of q20, xshift=-1.5cm] {$\mathrm{d}g_1$};
\node[state] (q22) [below  =1cm of q20] {$\mathrm{d}g_2$};
\node[state] (q23) [below  =1cm of q20, xshift=1.5cm] {$\mathrm{d}g_3$};
\node[state] (q24) [below  =1cm of q22] {$\mathrm{d}x$};

\node[state] (ar1) [above  =0.4cm of q0]{$\big\Downarrow$};
\node[state] (ar2) [above  =0.4cm of q10]{$\big\Downarrow$};
\node[state] (ar3) [above  =0.4cm of q20]{$\big\Downarrow$};

\path[->] 
(q0) edge (q1)
edge (q2)
edge (q3)
(q1) edge (q4)
(q2) edge (q5)
(q3) edge (q6)

(q10) edge (q11)
edge (q12)
edge (q13)
(q11) edge (q14)
(q12) edge (q14)
(q13) edge (q14)
(q14) edge[loop below,font=\scriptsize] (q14)

(q20) edge (q21)
edge (q22)
edge (q23)
(q21) edge (q24)
(q22) edge (q24)
(q23) edge (q24)
;

\end{tikzpicture}
	}
\caption{}\label{fig:graph-factoring-terminal}
\end{figure}

Notice that we removed the $\frac{\mathrm d x}{\mathrm d x}$ self loop in the end.
The parentheses `()' operator, 
it sees the inside differentiation expression as a whole and 
evaluate it before any adjacent outside operation. 
Because we see $x$ as the same variable during the process of factoring out the common factor, 
naturally it merges all the $\mathrm d x$ nodes to evaluate.

We notice that there is close correspondence between
differentiation graph and differentiation expression
which can be converted to each other.
Although differentiation graph and differentiation expression basically
are the same thing in different form, there are some differences between them.
The most obvious one is the terminal variable.
In algebraic differentiation expression, 
we use the same symbol in different position to denote they are the same,
whereas in differentiation graph, 
the terminal variable can be in different positions, 
e.g., the trees in \cref*{fig:graph-factoring} and \cref*{fig:graph-factoring-terminal}, 
they can also be in the same position, 
e.g., the DAG's in \cref*{fig:graph-factoring} and \cref*{fig:graph-factoring-terminal}.

Note how we arrived at the rooted DAG in \cref*{fig:graph-factoring}, 
we see the rooted tree as an intermediate graph, 
just like how we manipulate the function expression: 
first construct the tree hierarchy, 
then factor out the common factor by merging vertices.
The rooted tree is more similar to the algebraic expression
in terms of using same variable in different position.
e.g., the trees in \cref*{fig:graph-factoring} and \cref*{fig:graph-factoring-terminal}
have the exactly same structures as the corresponding algebraic expression.

So what does $\mathrm{d}y \to \mathrm{d}x$ mean?
In fact, we can further simplify the notation by removing `$\mathrm{d}$': $y \to x$.
The arrow $\to$ denotes that variable $y$ depends on variable $x$, 
and we take the derivative of $y$ with respect to $x$.Therefore,

\begin{figure}[H] 
	\centering
	\resizebox{0.66\textwidth}{!}{
    \begin{subfigure}[b]{0.3\textwidth}
    	\centering
    	\begin{tikzpicture}[>=stealth,thick]
          \usetikzlibrary{calc}
          \tikzstyle{state}=[]
          \node[state] (q0) {$\mathrm{d}f$};
          \node[state] (q1) [below  =0.95cm of q0, xshift=-1.5cm] {$\mathrm{d}g_1$};
          \node[state] (q2) [below  =0.95cm of q0] {$\mathrm{d}g_2$};
          \node[state] (q3) [below  =0.95cm of q0, xshift=1.5cm] {$\mathrm{d}g_3$};
          \node[state] (q4) [below  =0.95cm of q2] {$\mathrm{d}h$};
          \node[state] (q5) [below  =0.95cm of q4] {$\mathrm{d}x$};
          \path[->] 
          (q0) edge[] (q1)
          edge (q2)
          edge (q3)
          (q1) edge (q4)
          (q2) edge (q4)
          (q3) edge (q4)
          (q4) edge (q5)
          ;
       \end{tikzpicture}
	   \caption{}
    \end{subfigure}
    \begin{subfigure}[b]{0.039\textwidth}
    	\centering
    	\begin{tikzpicture}
    	   \node (q0) {};
    	   \node (q1)[above  =3.5cm of q0, xshift=0cm] {$\xRightarrow{\quad}$};
    	\end{tikzpicture}
    \end{subfigure}
    \begin{subfigure}[b]{0.3\textwidth}
        \centering
        	\begin{tikzpicture}[>=stealth,thick]
        	\usetikzlibrary{calc}
        	\tikzstyle{state}=[]
            \node[state] (q20) [xshift=0cm] {$f$};
            \node[state] (q21) [below  =1cm of q20, xshift=-1.5cm] {$g_1$};
            \node[state] (q22) [below  =1cm of q20] {$g_2$};
            \node[state] (q23) [below  =1cm of q20, xshift=1.5cm] {$g_3$};
            \node[state] (q24) [below  =1cm of q22] {$h$};
            \node[state] (q25) [below  =1cm of q24] {$x$};
        	\path[->] 
            (q20) edge node[auto,swap,font=\scriptsize]{$\frac{\partial f}{\partial g_1}$} (q21)
            edge node[font=\scriptsize]{$\frac{\partial f}{\partial g_2}$} (q22)
            edge node[auto,font=\scriptsize]{$\frac{\partial f}{\partial g_3}$} (q23)
            (q21) edge node[auto,swap,font=\scriptsize]{$\frac{\mathrm{d} g_1}{\mathrm{d} h}$} (q24)
            (q22) edge node[font=\scriptsize]{$\frac{\mathrm{d} g_2}{\mathrm{d} h}$} (q24)
            (q23) edge node[auto,font=\scriptsize]{$\frac{\mathrm{d} g_3}{\mathrm{d} h}$} (q24)
            (q24) edge node[font=\scriptsize]{$\frac{\mathrm{d} h}{\mathrm{d} x}$} (q25)
        	;
        	\end{tikzpicture}
    	\caption{}
    \end{subfigure}
    \begin{subfigure}[b]{0.039\textwidth}
	  \centering
	  \begin{tikzpicture}
	    \node (q0) {};
	    \node (q1)[above  =3.5cm of q0, xshift=0cm] {$\xRightarrow{\quad}$};
	  \end{tikzpicture}
    \end{subfigure}
    \begin{subfigure}[b]{0.3\textwidth}
	\centering
	\begin{tikzpicture}[>=stealth,thick]
	\usetikzlibrary{calc}
	\tikzstyle{state}=[]
    \node[state] (q30) [xshift=0cm] {$f$};
    \node[state] (q31) [below  =1cm of q30, xshift=-1.5cm] {$g_1$};
    \node[state] (q32) [below  =1cm of q30] {$g_2$};
    \node[state] (q33) [below  =1cm of q30, xshift=1.5cm] {$g_3$};
    \node[state] (q34) [below  =1cm of q32] {$h$};
    \node[state] (q35) [below  =1cm of q34] {$x$};
	\path[->] 
    (q30) edge (q31)
          edge (q32)
          edge (q33)
    (q31) edge (q34)
    (q32) edge (q34)
    (q33) edge (q34)
    (q34) edge (q35)
	;
	\end{tikzpicture}
	\caption{}
    \end{subfigure}
	}
\caption{Remove `$\mathrm{d}$'}\label{fig:graph-remove-d}
\end{figure}

We find that the trace of multi-path chain rule 
is the trace of dependence in the function expression.
Therefore a differentiation graph is also a dependence graph of function expressions. 
For differentiation graph,
if we define the direction of edges in paths as vertical, 
the direction of aligning sibling paths is horizontal, 
according to multi-path chain rule, 
between edges in the vertical direction there are multiplication operators, 
between sibling paths aligned in the horizontal direction 
there are addition operators. 
However, we cannot compute the root value or intermediate value of function expression
using its dependence graph, as there is no operator in dependence graph.
 
Note that if we reverse the direction of all the edges in \cref*{fig:graph-remove-d} (c), 
it becomes a computational graph. 
Computational graph is more like a generalized term, 
which means it is not limited to the computation of derivatives. 
Differentiation graph is computational graph, 
just like we say differentiation expression is algebraic expression. 
As Parr and Howard have mentioned in their paper(\cite{parr2018matrix},2018), 
the opposite arrow direction denotes the direction of data flow. 
In the diagrams of following sections, 
we will still use the direction of dependence as the default direction. 
\\

\section{Notations and Terminology}

Let $G=(V,E)$ be a differentiation graph, $v_i\in V \ni v_j\cap v_i\prec v_j$. 
We use $P_{v_j\to v_i}$ to denote the set of all the paths from vertex $v_j$ to $v_i$ 
and $p_{v_j\to v_i} (\in P_{v_j\to v_i})$ to denote a path in $P_{v_j\to v_i}$.
$|P_{v_j\to v_i}|$ is the number of paths in $P_{v_j\to v_i}$, 
$|p_{v_j\to v_i}|$ is the length of path $p_{v_j\to v_i}$, i.e., the number of edges in the path;
Then the formula of multi-path chain rule or Bauer’s formula can be rewritten as:

$$
\frac{\partial v_y}{\partial v_x}=\sum_{p_j\in \mathrm{P}_{v_y\to v_x}}\prod_{e_i\in p_j}e_i
$$

We use $e_i$ to denote the $i$th edge in $E$, 
and $e_i^{-}$ and $e_i^{+}$ to denote 
the source node and destination node of directed edge $e_i$ respectively. 
We also use $\langle j,i\rangle$ or $\langle v_j,v_i\rangle$ to denote the directed edge from vertex $v_j$ to $v_i$. 

We use $P_{e_j\to e_i}$ to denote 
the set of all the paths from edge $e_j$ to $e_i$.
Apparently, $P_{e_j\to e_i}\equiv P_{e_j^{+}\to e_i^{-}}$.
We use $v_i^{-}$ and $v_i^{+}$ to denote 
the incoming edges and outgoing edges of intermediate vertex $v_i$ respectively.
Therefore the set of all adjacent edge paths between the incoming edges and outgoing edges of $v_i$
is $P_{v_i}\equiv P_{v_i^{-}\to v_i^{+}}$. 
$\forall p_k\in P_{v_i}, |p_k|=2$.
We use $v_i^{--}$ or $v_i^{=}$ to denote all the predecessors of $v_i$, 
and $v_i^{++}$ or $v_i^{\#}$ to denote all the successors of $v_i$. 
Let $V'\subseteq V$, 
then $V'^{=}$ is the set of all the predecessors of vertices in $V'$, 
$V'^{\#}$ is the set of all the successors of vertices in $V'$.

Let $Y$,$Z$,$X$ be the roots, intermediates and terminals respectively,
$Y\cup Z\cup X=V$.
Then $P_{Y\to X}$ is the set of all the paths from roots $\mathrm{Y}$ to terminals $X$. 
We use $Y^{+}$ to denote all the outgoing edges from $Y$, 
$Y^{\#}$ to denote all the successors of roots,
$X^{-}$ to denote all the incoming edges to $X$, 
$X^{=}$ to denote all the predecessors of terminals. 
$Y^{-}=\emptyset$, $Y^{=}=\emptyset$, $X^{+}=\emptyset$, $X^{\#}=\emptyset$. 
$\forall v_i\in V$: if $v_i^{-}=\emptyset \cup v_i^{=}=\emptyset$, then $v_i\in Y$; 
if $v_i^{+}=\emptyset \cup v_i^{\#}=\emptyset$, then $v_i\in X$;

We use $\mathrm{D}^{-}(v_i)$ to denote the in-degree of $v_i$, 
i.e., the number of incoming edges of $v_i$, 
$\mathrm{D}^{+}(v_i)$ to denote the out-degree of $v_i$, 
i.e., the number of outgoing edges of $v_i$. 
$\forall x_i\in X,\mathrm{D}^{+}(x_i)=0$, 
$\forall y_j\in Y,\mathrm{D}^{-}(y_j)=0$.
$\forall v_i\in V, \mathrm{D}^{-}(v_i)\ge 0, \mathrm{D}^{+}(v_i)\ge 0$. 
We use $\mathrm{D}_r^{-}(v_i)$ to denote the r-degree of $v_i$, i.e., 
the number of root vertices that have at least one path to $v_i$, 
and $\mathrm{D}_t^{+}(v_i)$ to denote the t-degree of $v_i$, i.e., 
the number of terminal vertices to which $v_i$ has at least one path. 
We use $X_i$ to denote the set of terminals to which $v_i$ has at least one path, 
$Y_i$ to denote the set of roots that have at least one path to $v_i$. 
Note that $X_i$ and $Y_i$ are the concepts of index domain and index range 
mentioned in \cite{griewank2008} (P147). 
$\mathrm{D}_r^{-}(v_i)=|Y_i|$,  $\mathrm{D}_t^{+}(v_i)=|X_i|$. 

Notice that there could be more than one path going through one edge. 
We use $\mathrm{O}_{P}(e_{k})$ to denote 
the overlapping degree of an edge $e_{k}\in E$ in a set of paths $P$,
i.e., the number of paths in $P$ going through edge $e_{k}$.
Therefore, $\forall v_i\in Z$,

$$
\mathrm{O}_{P_{v_i}}(e_{k})=
\begin{cases}
\mathrm{D}^{+}(v_i)  &\quad (e_{k}\in v_i^{-}) \\
\mathrm{D}^{-}(v_i)  &\quad (e_{k}\in v_i^{+})
\end{cases}
$$

$\forall e_k\in E, \mathrm{O}_{P}(e_{k})\ge 0$.

\begin{definition}\label{def:block}
In a differentiation graph $G=(V,E)$, 
if there are at least two separate paths between vertex $v_i$ and $v_j$ 
that do not share any intermediate vertex at the same graph location, 
and all the intermediate vertices in $P_{v_i\to v_j}$ have no outgoing edge to 
or incoming edge from vertices that are not in $P_{v_i\to v_j}$, 
then all the vertices(including $v_i$ and $v_j$) and edges 
in all those paths from $v_i$ to $v_j$ form an block $b\langle v_i,v_j\rangle$. 
$v_i$ is the source of $b\langle v_i,v_j\rangle$, 
$v_j$ is the sink of $b\langle v_i,v_j\rangle$. 
\end{definition}

If there is a block between $v_i$ and $v_j$, 
it is a differentiation graph of $\frac{\mathrm{d}v_i}{\mathrm{d}v_j}$. 

\begin{definition}\label{def:simple-chain}
Assume that there are two different vertices $v_i$ and $v_j$ in differentiation graph $G=(V,E)$. 
A path $p_{v_i\to v_j}$ is a direct chain or direct simple chain $e'\langle v_i,v_j\rangle$ 
if it has at least 1 intermediate vertex and the in-degree 
and out-degree of every intermediate vertex in $p_{v_i\to v_j}$ are both 1. 
If after every block whose intermediate vertices do not include $v_i$ or $v_j$ 
is replaced by a substitute edge in $G$, 
path $p_{v_i\to v_j}$ is a direct simple chain that contains at least 1 substitute edge, 
then all the non-substitute edges and all the vertices in $p_{v_i\to v_j}$ 
and all the vertices and edges in the corresponding blocks of those substitute edges 
in $p_{v_i\to v_j}$ form an indirect chain $c\langle v_i,v_j\rangle$. 
Both direct chain and indirect chain are called chain. 
If all the corresponding blocks of those substitute edges in $p_{v_i\to v_j}$ are simple blocks, 
then $c\langle v_i,v_j\rangle$ is an indirect simple chain $e'\langle v_i,v_j\rangle$. 
Both direct simple chain and indirect simple chain are called simple chain. 
If a chain $c\langle v_i,v_j\rangle$ is not an simple chain, then it is a complex chain. 
$v_i$ is the source of $c\langle v_i,v_j\rangle$ and $e'\langle v_i,v_j\rangle$, 
$v_j$ is the sink of $c\langle v_i,v_j\rangle$ and $e'\langle v_i,v_j\rangle$. 
\end{definition}

Note that an edge could be a path of length 1, 
but here in our definition it is not a chain, 
and an indirect chain has more than one path. 

\begin{definition}\label{def:simple-block} 
In a differentiation graph $G=(V,E)$, 
if every outgoing edge of vertex $v_i$ belongs to a simple chain to $v_j$ or connects to $v_j$, 
and there are at least 2 outgoing edge of $v_i$ are not in the same simple chain, 
then all the vertices(including $v_i$ and $v_j$) and edges in all those paths from $v_i$ to $v_j$ 
form an simple block $e\langle v_i,v_j\rangle$. 
$v_i$ is the source of $e\langle v_i,v_j\rangle$, 
$v_j$ is the sink of $e\langle v_i,v_j\rangle$. 
If all the simple chains in $e\langle v_i,v_j\rangle$ are direct simple chains, then
$e\langle v_i,v_j\rangle$ is a direct simple block, otherwise it is an indirect simple block.
If a block $b\langle v_i,v_j\rangle$ is not an simple block, then it is a complex block.
\end{definition}

\newpage
\begin{definition}\label{def:vertex-depth-level} 
In a differentiation graph $G=(V,E)$, 
all the root vertices are in the same first level of depth, 
and all the terminal vertices are in the same last level of depth. 
The depth level of an intermediate vertex is 
the length of the longest path of all the paths from root vertices to that vertex. 
We use $L(v_i)$ to denote the depth level of vertex $v_i$. 
If vertex $v_i$ has an directed edge $e$ connected to vertex $v_j$ and $L(v_j)-L(v_i)>1$, 
then $e$ is a cross-level edge. 
We use $L(Y)$ to denote the first level of depth, 
$L(X)$ to denote the last level of depth, 
and $L(G)$ to denote the depth of $G$. 
$L(Y)=0$, $L(G)=L(X)$. 
\end{definition} 

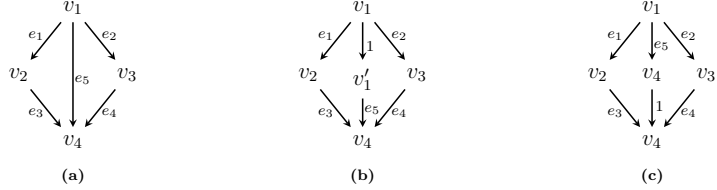
\begin{figure}[H] 
	\centering
	\resizebox{.22\textwidth}{!}{
		\begin{subfigure}[b]{0.3\textwidth}
			\centering
			\begin{tikzpicture}[>=stealth,thick]
			\usetikzlibrary{calc}
			\tikzstyle{state}=[]
			\node[state] (q1) [below  =0.7cm of q0, xshift=0cm] {$v_1$};
			\node[state] (q2) [below  =0.7cm of q1, xshift=-1cm] {$v_2$};
			\node[state] (q3) [below  =0.7cm of q1, xshift=1cm] {$v_3$};
			\node[state] (q4) [below  =0.7cm of q3, xshift=-1cm] {$v_4$};
			
			\path[->] 
			
			(q1) edge node[above=-4pt, xshift=-5pt, font=\scriptsize]{$e_1$} (q2)
			edge node[above=-4pt, xshift=5pt, font=\scriptsize]{$e_2$} (q3)
			edge node[below=-4pt, xshift=5pt, font=\scriptsize]{$e_5$} (q4)
			(q2) edge node[below=-4pt, xshift=-5pt, font=\scriptsize]{$e_3$} (q4)
			(q3) edge node[below=-4pt, xshift=5pt, font=\scriptsize]{$e_4$} (q4)
			;
			\end{tikzpicture}
			\caption{}
		\end{subfigure}
	}
	\resizebox{.22\textwidth}{!}{
		\begin{subfigure}[b]{0.3\textwidth}
			\centering
			\begin{tikzpicture}[>=stealth,thick]
			\usetikzlibrary{calc}
			\tikzstyle{state}=[]
			\node[state] (q1) [below  =0.7cm of q0, xshift=0cm] {$v_1$};
			\node[state] (q1-1) [below  =0.7cm of q1, xshift=0cm] {$v_1'$};
			\node[state] (q2) [below  =0.7cm of q1, xshift=-1cm] {$v_2$};
			\node[state] (q3) [below  =0.7cm of q1, xshift=1cm] {$v_3$};
			\node[state] (q4) [below  =0.7cm of q3, xshift=-1cm] {$v_4$};
			
			\path[->] 
			
			(q1) edge node[above=-4pt, xshift=-5pt, font=\scriptsize]{$e_1$} (q2)
			edge node[above=-4pt, xshift=5pt, font=\scriptsize]{$e_2$} (q3)
			edge node[below=-4pt, xshift=4pt, font=\scriptsize]{$1$} (q1-1)
			(q1-1) edge node[below=-8pt, xshift=5pt, font=\scriptsize]{$e_5$} (q4)
			(q2) edge node[below=-4pt, xshift=-5pt, font=\scriptsize]{$e_3$} (q4)
			(q3) edge node[below=-4pt, xshift=5pt, font=\scriptsize]{$e_4$} (q4)
			;
			\end{tikzpicture}
			\caption{}
		\end{subfigure}
	}
	\resizebox{.22\textwidth}{!}{
		\begin{subfigure}[b]{0.3\textwidth}
			\centering
			\begin{tikzpicture}[>=stealth,thick]
			\usetikzlibrary{calc}
			\tikzstyle{state}=[]
			\node[state] (q1) [below  =0.7cm of q0, xshift=0cm] {$v_1$};
			\node[state] (q1-1) [below  =0.7cm of q1, xshift=0cm] {$v_4$};
			\node[state] (q2) [below  =0.7cm of q1, xshift=-1cm] {$v_2$};
			\node[state] (q3) [below  =0.7cm of q1, xshift=1cm] {$v_3$};
			\node[state] (q4) [below  =0.7cm of q3, xshift=-1cm] {$v_4$};
			
			\path[->] 
			
			(q1) edge node[above=-4pt, xshift=-5pt, font=\scriptsize]{$e_1$} (q2)
			edge node[above=-4pt, xshift=5pt, font=\scriptsize]{$e_2$} (q3)
			edge node[below=-4pt, xshift=5pt, font=\scriptsize]{$e_5$} (q1-1)
			(q1-1) edge node[below=-8pt, xshift=4pt, font=\scriptsize]{$1$} (q4)
			(q2) edge node[below=-4pt, xshift=-5pt, font=\scriptsize]{$e_3$} (q4)
			(q3) edge node[below=-4pt, xshift=5pt, font=\scriptsize]{$e_4$} (q4)
			;
			\end{tikzpicture}
			\caption{}
		\end{subfigure}
	}
	\caption{Cross-level edge segmentation}\label{fig:cross-level-edge-segmentation}
\end{figure}

We can segment a cross-level edge by filling the positions of depth levels 
crossed by it with vertices derived from its head or tail vertex, 
turning it into a direct simple chain. 
e.g., \cref*{fig:cross-level-edge-segmentation} (b) or (c) is the segmentation result of 
cross-level edge $\langle v_1,v_4\rangle$ in \cref*{fig:cross-level-edge-segmentation} (a). 
In a differentiation graph in which all the cross-level edges are segmented, 
every two adjacent depth levels form a bipartite subgraph or local Jacobian. 
Note that the calculation of $1\cdot e_5$ in $(e_1e_3+1\cdot e_5+e_2e_4)$ is unnecessary, 
just like the calculation of zero entries in local Jacobian multiplication. 

We denote the set of all vertices in depth level $l$ by $V^l$, 
e.g., $Y=V^0$, $X=V^h$, where $h$ is the depth level of terminals. 
Let $V^a$, $V^b$ and $V^c$($a<b<c$) be the sets of vertices in different depth levels, 
then for $v_i\in V^b$, 
the set of vertices in $V^c$ to which $v_i$ has at least one path is $V^c_i$ or $X^c_i$, 
the set of vertices in $V^a$ that have at least one path to $v_i$ is $V^a_i$ or $Y^a_i$, 
where $X_i=V^h_i=X^h_i$, $Y_i=V^0_i=Y^0_i$. 

We use $J^i_k$ to denote the Jacobian between depth level $i$ and $k$, 
e.g., the global Jacobian is $J^0_{L(x)}$. 
However, local Jacobians are not limited to two different depth levels. 
We use $J\langle S_1|S_2\rangle$ to denote the Jacobian between two sets of vertices, 
e.g., $J\langle Y|X\rangle$ is the global Jacobian. 
We use $J\langle v_a,v_b|v_c,v_d\rangle$ to denote the local Jacobian 
between $v_a,v_b$ and $v_c,v_d$, 
or $J\langle v_i|v_i^{\#}\rangle$ to denote the local Jacobian 
between $v_i$ and $v_i^{\#}$.

\begin{definition} 
Let $G'=(U,V,E')$ be a bipartite graph or subgraph 
with directed edges from $U$ to $V$, 
after we remove all the directions of directed edges in $E'$, 
it becomes an undirected graph $G=(V,E)$. 
If $\exists v_i\in G$, 
$\forall v_j\in G\cap v_j\ne v_i, P_{v_i\to v_j}\ne\emptyset$, 
then $G'$ is an undivided bipartite graph or subgraph. 
\end{definition}

\section{Factorization and Conversion}

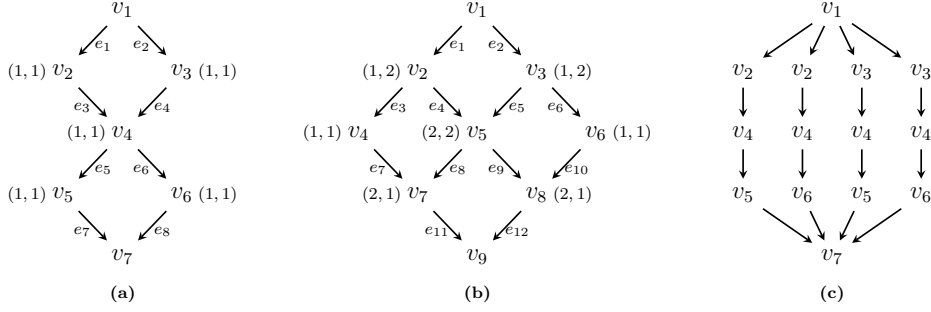
\begin{figure}[H] 
	\centering
	\resizebox{.8\textwidth}{!}{
		\begin{subfigure}[b]{0.3\textwidth}
			\centering
			\begin{tikzpicture}[>=stealth,thick]
			\usetikzlibrary{calc}
			\tikzstyle{state}=[]
			
			\node[state] (q1) [below  =0.7cm of q0, xshift=0cm] {$v_1$};
			\node[state] (q2) [below  =0.5cm of q1, xshift=-1cm] {$v_2$};
			\node[state] (lq2) [left  =-6pt of q2, xshift=-0pt, font=\scriptsize] {$(1,1)$};
			\node[state] (q3) [below  =0.5cm of q1, xshift=1cm] {$v_3$};
			\node[state] (lq3) [right  =-6pt of q3, xshift=-0pt, font=\scriptsize] {$(1,1)$};
			\node[state] (q4) [below  =0.5cm of q3, xshift=-1cm] {$v_4$};
			\node[state] (lq4) [left  =-6pt of q4, xshift=-0pt, font=\scriptsize] {$(1,1)$};
			\node[state] (q5) [below  =0.5cm of q4, xshift=-1cm] {$v_5$};
			\node[state] (lq5) [left  =-6pt of q5, xshift=-0pt, font=\scriptsize] {$(1,1)$};
			\node[state] (q6) [below  =0.5cm of q4, xshift=1cm] {$v_6$};
			\node[state] (lq6) [right  =-6pt of q6, xshift=-0pt, font=\scriptsize] {$(1,1)$};
			\node[state] (q7) [below  =0.5cm of q6, xshift=-1cm] {$v_7$};
			
			\path[->] 
			
			(q1) edge node[below=-4pt, xshift=5pt, font=\scriptsize]{$e_1$} (q2)
			     edge node[below=-4pt, xshift=-5pt, font=\scriptsize]{$e_2$} (q3)
			(q2) edge node[below=-4pt, xshift=-5pt, font=\scriptsize]{$e_3$} (q4)
			(q3) edge node[below=-4pt, xshift=5pt, font=\scriptsize]{$e_4$} (q4)
			(q4) edge node[below=-4pt, xshift=5pt, font=\scriptsize]{$e_5$} (q5)
			     edge node[below=-4pt, xshift=-5pt, font=\scriptsize]{$e_6$} (q6)
			(q5) edge node[below=-4pt, xshift=-5pt, font=\scriptsize]{$e_7$} (q7)
			(q6) edge node[below=-4pt, xshift=5pt, font=\scriptsize]{$e_8$} (q7)
			;
			\end{tikzpicture}
			\caption{}
		\end{subfigure}
		\begin{subfigure}[b]{0.4\textwidth}
			\centering
			\begin{tikzpicture}[>=stealth,thick]
			\usetikzlibrary{calc}
			\tikzstyle{state}=[]
			
			\node[state] (q1) [below  =0.7cm of q0, xshift=0cm] {$v_1$};
			\node[state] (q2) [below  =0.5cm of q1, xshift=-1cm] {$v_2$};
			\node[state] (lq2) [left  =-6pt of q2, xshift=-0pt, font=\scriptsize] {$(1,2)$};
			\node[state] (q3) [below  =0.5cm of q1, xshift=1cm] {$v_3$};
			\node[state] (lq3) [right  =-6pt of q3, xshift=-0pt, font=\scriptsize] {$(1,2)$};
			\node[state] (q4) [below  =0.5cm of q2, xshift=-1cm] {$v_4$};
			\node[state] (lq4) [left  =-6pt of q4, xshift=-0pt, font=\scriptsize] {$(1,1)$};
			\node[state] (q5) [below  =0.5cm of q3, xshift=-1cm] {$v_5$};
			\node[state] (lq5) [left  =-6pt of q5, xshift=-0pt, font=\scriptsize] {$(2,2)$};
			\node[state] (q6) [below  =0.5cm of q3, xshift=1cm] {$v_6$};
			\node[state] (lq6) [right  =-6pt of q6, xshift=-0pt, font=\scriptsize] {$(1,1)$};
			\node[state] (q7) [below  =0.5cm of q5, xshift=-1cm] {$v_7$};
			\node[state] (lq7) [left  =-6pt of q7, xshift=-0pt, font=\scriptsize] {$(2,1)$};
			\node[state] (q8) [below  =0.5cm of q5, xshift=1cm] {$v_8$};
			\node[state] (lq8) [right  =-6pt of q8, xshift=-0pt, font=\scriptsize] {$(2,1)$};
			\node[state] (q9) [below  =0.5cm of q8, xshift=-1cm] {$v_9$};
			
			\path[->] 
			(q1) edge node[below=-4pt, xshift=5pt, font=\scriptsize]{$e_1$} (q2)
			     edge node[below=-4pt, xshift=-5pt, font=\scriptsize]{$e_2$} (q3)
			(q2) edge node[below=-4pt, xshift=5pt, font=\scriptsize]{$e_3$} (q4)
			     edge node[below=-4pt, xshift=-5pt, font=\scriptsize]{$e_4$} (q5)
			(q3) edge node[below=-4pt, xshift=5pt, font=\scriptsize]{$e_5$} (q5)
			     edge node[below=-4pt, xshift=-5pt, font=\scriptsize]{$e_6$} (q6)
			(q4) edge node[below=-4pt, xshift=-5pt, font=\scriptsize]{$e_7$} (q7)
			(q5) edge node[below=-4pt, xshift=5pt, font=\scriptsize]{$e_8$} (q7)
			     edge node[below=-4pt, xshift=-5pt, font=\scriptsize]{$e_9$} (q8)
			(q6) edge node[below=-4pt, xshift=5pt, font=\scriptsize]{$e_{10}$} (q8)
			(q7) edge node[below=-4pt, xshift=-5pt, font=\scriptsize]{$e_{11}$} (q9)
			(q8) edge node[below=-4pt, xshift=5pt, font=\scriptsize]{$e_{12}$} (q9)
			;
			\end{tikzpicture}
			\caption{}
		\end{subfigure}
		\begin{subfigure}[b]{0.3\textwidth}
			\centering
			\begin{tikzpicture}[>=stealth,thick]
			\usetikzlibrary{calc}
			\tikzstyle{state}=[]
			
			\node[state] (q1) [below  =0.7cm of q0, xshift=0cm] {$v_1$};
			\node[state] (q2-1) [below  =0.5cm of q1, xshift=-1.5cm] {$v_2$};
			\node[state] (q2-2) [below  =0.5cm of q1, xshift=-.5cm] {$v_2$};
			\node[state] (q3-1) [below  =0.5cm of q1, xshift=.5cm] {$v_3$};
			\node[state] (q3-2) [below  =0.5cm of q1, xshift=1.5cm] {$v_3$};
			
			\node[state] (q4-1) [below  =0.5cm of q2-1, xshift=0cm] {$v_4$};
			\node[state] (q4-2) [below  =0.5cm of q2-2, xshift=0cm] {$v_4$};
			\node[state] (q4-3) [below  =0.5cm of q3-1, xshift=0cm] {$v_4$};
			\node[state] (q4-4) [below  =0.5cm of q3-2, xshift=0cm] {$v_4$};
			
			\node[state] (q5-1) [below  =0.5cm of q4-1, xshift=0cm] {$v_5$};
			\node[state] (q5-2) [below  =0.5cm of q4-3, xshift=0cm] {$v_5$};
			\node[state] (q6-1) [below  =0.5cm of q4-2, xshift=0cm] {$v_6$};
			\node[state] (q6-2) [below  =0.5cm of q4-4, xshift=0cm] {$v_6$};
			
			\node[state] (q7) [below  =0.5cm of q5-2, xshift=-.5cm] {$v_7$};
			
			\path[->] 
			(q1) edge (q2-1)
			edge (q2-2)
			edge (q3-1)
			edge (q3-2)
			
			(q2-1) edge (q4-1)
			(q2-2) edge (q4-2)
			(q3-1) edge (q4-3)
			(q3-2) edge (q4-4)
			
			(q4-1) edge (q5-1)
			(q4-2) edge (q6-1)
			(q4-3) edge (q5-2)
			(q4-4) edge (q6-2)
			
			(q5-1) edge (q7)
			(q5-2) edge (q7)
			(q6-1) edge (q7)
			(q6-2) edge (q7)
			
			;
			\end{tikzpicture}
			\caption{}
		\end{subfigure}

	}
	\caption{Simple Block and Complex Block}
	\label{fig:complex-dag-and-complete-path-dag}
\end{figure}

The conversion between differentiation graph and algebraic expression
is obvious when the graph is simple. 
e.g., in \cref*{fig:complex-dag-and-complete-path-dag} (a)
there are two simple blocks $e\langle v_1,v_4\rangle$ and $e\langle v_4,v_7\rangle$ chained together, 
we can immediately write down its algebraic expression:

\begin{equation}\label{eq-a-simple}
\frac{\mathrm{d} v_1}{\mathrm{d} v_7}=
(e_1e_3+e_2e_4)(e_5e_7+e_6e_8) 
\end{equation}
\\
\noindent
where the edge label $e$ denotes the derivative of its source node 
with respect to its destination node, 
e.g., $e_1=\frac{\partial v_1}{\partial v_2}$.
We notice that a simple block in differentiation graph can be converted to 
a subexpression block with parenthesis in algebraic expression.

However, it is difficult to get an algebraic expression in 
\cref*{fig:complex-dag-and-complete-path-dag} (b) at one glance,
because it is a complex block.
Differentiation graph is more flexible 
when it comes to the same intermediate variables(vertices) or derivatives(edges), 
it can always merge them, whereas in a single algebraic expression, 
due to its limitation of expressivity, 
there is no direct representation or corresponding counterpart of complex block.

One solution for converting a complex differentiation graph to an algebraic expression
is to transform it into a simple differentiation graph with simple chains and simple blocks,
which can be converted to algebraic expression directly.
e.g., \cref*{fig:complex-dag-and-complete-path-dag} (c) 
is a direct simple block that lists all paths from root to terminal
in \cref*{fig:complex-dag-and-complete-path-dag} (a).
Its algebraic expression would be like this:
\\
\begin{equation}\label{eq-a-full}
\frac{\mathrm{d} v_1}{\mathrm{d} v_7}=
\left(
e_1e_3e_5e_7+
e_1e_3e_6e_8+
e_2e_4e_5e_7+
e_2e_4e_6e_8
\right)
\end{equation}
\\
\noindent
But there is a difference between the result algebraic expressions.
Note that it takes 12 fma's to compute \cref*{eq-a-full}
while only 5 fma's for \cref*{eq-a-simple}.

There are two ways to transform \cref*{fig:complex-dag-and-complete-path-dag} (b).
One way is that first transform it to a direct simple block 
that lists all the paths from root to terminal,
then merge these vertices and edges again, 
making sure only simple blocks or simple chains are generated.
The other way is to break down the overly merged graph into 
a more convertible one that has no complex blocks.
e.g., \cref*{fig:2-step-factorization} demonstrates how
we factorize \cref*{fig:complex-dag-and-complete-path-dag} (b).

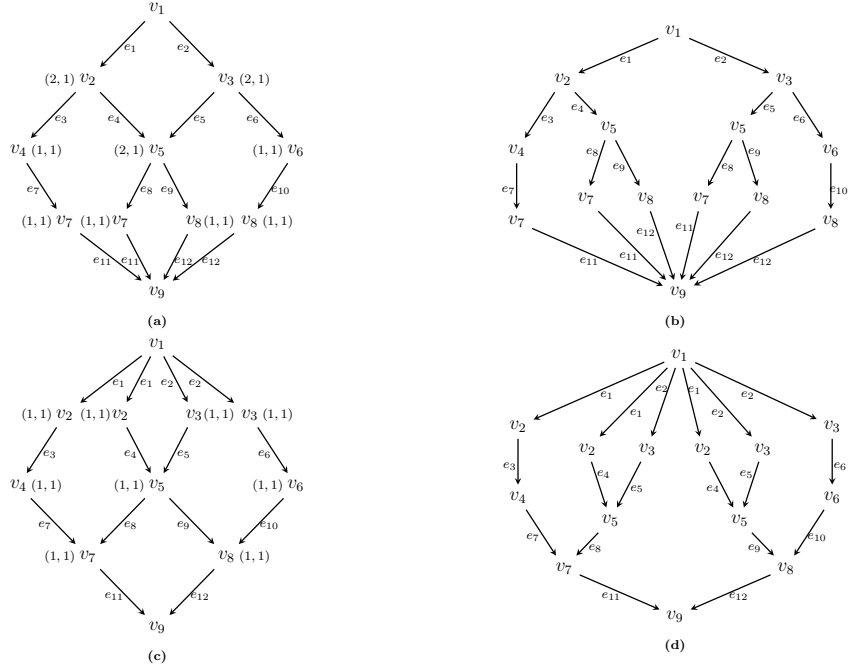
\begin{figure}[H] 
	\centering
        \resizebox{0.4\textwidth}{!}{
		\begin{subfigure}[b]{0.65\textwidth}
			\centering
			\begin{tikzpicture}[>=stealth,thick]
			\usetikzlibrary{calc}
			\tikzstyle{state}=[]
			
			\node[state] (q1) [xshift=0cm] {$v_1$};
			\node[state] (q2) [below  =1cm of q1, xshift=-1.5cm] {$v_2$};
			\node[state] (lq2) [left  =-6pt of q2, xshift=-0pt, font=\scriptsize] {$(2,1)$};
			\node[state] (q3) [below  =1cm of q1, xshift=1.5cm] {$v_3$};
			\node[state] (lq3) [right  =-6pt of q3, xshift=-0pt, font=\scriptsize] {$(2,1)$};
			\node[state] (q4) [below  =1cm of q2, xshift=-1.5cm] {$v_4$};
			\node[state] (lq4) [right  =-6pt of q4, xshift=-0pt, font=\scriptsize] {$(1,1)$};
			\node[state] (q5) [below  =1cm of q3, xshift=-1.5cm] {$v_5$};
			\node[state] (lq5) [left  =-6pt of q5, xshift=-0pt, font=\scriptsize] {$(2,1)$};
			\node[state] (q6) [below  =1cm of q3, xshift=1.5cm] {$v_6$};
			\node[state] (lq6) [left  =-6pt of q6, xshift=-0pt, font=\scriptsize] {$(1,1)$};
			\node[state] (q7-1) [below  =1cm of q5, xshift=-2.0cm] {$v_7$};
			\node[state] (lq7-1) [left  =-6pt of q7-1, xshift=-0pt, font=\scriptsize] {$(1,1)$};
			\node[state] (q7-2) [below  =1cm of q5, xshift=-.8cm] {$v_7$};
			\node[state] (lq7-2) [left  =-8pt of q7-2, xshift=-0pt, font=\scriptsize] {$(1,1)$};
			\node[state] (q8-1) [below  =1cm of q5, xshift=2.0cm] {$v_8$};
			\node[state] (lq8-1) [right  =-6pt of q8-1, xshift=-0pt, font=\scriptsize] {$(1,1)$};
			\node[state] (q8-2) [below  =1cm of q5, xshift=.8cm] {$v_8$};
			\node[state] (lq8-2) [right  =-8pt of q8-2, xshift=-0pt, font=\scriptsize] {$(1,1)$};
			\node[state] (q9) [below  =1cm of q8-1, xshift=-2.0cm] {$v_9$};
			
			\path[->] 
			(q1) edge node[below=-4pt, xshift=5pt, font=\scriptsize]{$e_1$} (q2)
			     edge node[below=-4pt, xshift=-5pt, font=\scriptsize]{$e_2$} (q3)
			(q2) edge node[below=-4pt, xshift=5pt, font=\scriptsize]{$e_3$} (q4)
			     edge node[below=-4pt, xshift=-5pt, font=\scriptsize]{$e_4$} (q5)
			(q3) edge node[below=-4pt, xshift=5pt, font=\scriptsize]{$e_5$} (q5)
			     edge node[below=-4pt, xshift=-5pt, font=\scriptsize]{$e_6$} (q6)
			(q4) edge node[below=-4pt, xshift=-5pt, font=\scriptsize]{$e_7$} (q7-1)
			(q5) edge node[below=-4pt, xshift=5pt, font=\scriptsize]{$e_8$} (q7-2)
			     edge node[below=-4pt, xshift=-5pt, font=\scriptsize]{$e_9$} (q8-2)
			(q6) edge node[below=-4pt, xshift=5pt, font=\scriptsize]{$e_{10}$} (q8-1)
			(q7-1) edge node[below=-4pt, xshift=-5pt, font=\scriptsize]{$e_{11}$} (q9)
			(q8-1) edge node[below=-4pt, xshift=5pt, font=\scriptsize]{$e_{12}$} (q9)
			(q7-2) edge node[below=-4pt, xshift=-5pt, font=\scriptsize]{$e_{11}$} (q9)
            (q8-2) edge node[below=-4pt, xshift=5pt, font=\scriptsize]{$e_{12}$} (q9)
			;
			\end{tikzpicture}
            
			\caption{}
		\end{subfigure}
	    }
        \resizebox{0.4\textwidth}{!}{
		\begin{subfigure}[b]{0.65\textwidth}
			\centering
			\begin{tikzpicture}[>=stealth,thick]
			\usetikzlibrary{calc}
			\tikzstyle{state}=[]
			
			\node[state] (q1) [xshift=0cm] {$v_1$};
			\node[state] (q2) [below  =0.5cm of q1, xshift=-2.4cm] {$v_2$};
			\node[state] (q3) [below  =0.5cm of q1, xshift=2.4cm] {$v_3$};
			\node[state] (q4) [below  =1cm of q2, xshift=-1cm] {$v_4$};
			
			\node[state] (q5-1) [below  =0.5cm of q2, xshift=1cm] {$v_5$};
			\node[state] (q7-1) [below  =1.5cm of q5-1, xshift=-2cm] {$v_7$};
			\node[state] (q7-3) [below  =1cm of q5-1, xshift=-0.5cm] {$v_7$};
			\node[state] (q8-1) [below  =1cm of q5-1, xshift=0.8cm] {$v_8$};
			\node[state] (q9) [below  =1.5cm of q8-1, xshift=0.7cm] {$v_9$};
			
			\node[state] (q6) [below  =1cm of q3, xshift=1cm] {$v_6$};
			\node[state] (q5-2) [below  =0.5cm of q3, xshift=-1cm] {$v_5$};
			\node[state] (q7-2) [below  =1cm of q5-2, xshift=-0.8cm] {$v_7$};
			\node[state] (q8-2) [below  =1.5cm of q5-2, xshift=2cm] {$v_8$};
			\node[state] (q8-3) [below  =1cm of q5-2, xshift=0.5cm] {$v_8$};
			
			\path[->] 
			(q1) edge node[below=-4pt, xshift=5pt, font=\scriptsize]{$e_1$} (q2)
			     edge node[below=-4pt, xshift=-5pt, font=\scriptsize]{$e_2$} (q3)
			(q2) edge node[below=-4pt, xshift=5pt, font=\scriptsize]{$e_3$} (q4)
			     edge node[below=-4pt, xshift=-5pt, font=\scriptsize]{$e_4$} (q5-1)
			
			(q3) edge node[below=-4pt, xshift=-5pt, font=\scriptsize]{$e_6$} (q6)
			     edge node[below=-4pt, xshift=5pt, font=\scriptsize]{$e_5$} (q5-2)
			(q4) edge node[below=-4pt, xshift=-5pt, font=\scriptsize]{$e_7$} (q7-1)
			(q5-1) edge node[above=-0pt, xshift=-3pt, font=\scriptsize]{$e_8$} (q7-3)
			       edge node[below=-4pt, xshift=-5pt, font=\scriptsize]{$e_9$} (q8-1)
			(q7-1) edge node[below=-4pt, xshift=-5pt, font=\scriptsize]{$e_{11}$} (q9)
			(q7-3) edge node[below=-0pt, xshift=-2pt, font=\scriptsize]{$e_{11}$} (q9)
			(q8-1) edge node[below=-15pt, xshift=-10pt, font=\scriptsize]{$e_{12}$} (q9)
			
			(q6) edge node[below=-4pt, xshift=5pt, font=\scriptsize]{$e_{10}$} (q8-2)
			(q5-2) edge node[below=-4pt, xshift=5pt, font=\scriptsize]{$e_8$} (q7-2)
			       edge node[above=-0pt, xshift=3pt, font=\scriptsize]{$e_9$} (q8-3)
			(q7-2) edge node[above=5pt, xshift=-3pt, font=\scriptsize]{$e_{11}$} (q9)
			(q8-2) edge node[below=-4pt, xshift=5pt, font=\scriptsize]{$e_{12}$} (q9)
			(q8-3) edge node[below=-0pt, xshift=3pt, font=\scriptsize]{$e_{12}$} (q9)
			;
			\end{tikzpicture}
			\caption{}
		\end{subfigure}
	    }
		\resizebox{0.4\textwidth}{!}{
    	\begin{subfigure}[b!]{0.65\textwidth}
	    	\centering
	    	\begin{tikzpicture}[>=stealth,thick]
	    	\usetikzlibrary{calc}
	    	\tikzstyle{state}=[]
	    	
	    	\node[state] (q1) [xshift=0cm] {$v_1$};
	    	\node[state] (q2-1) [below  =1cm of q1, xshift=-2cm] {$v_2$};
	    	\node[state] (lq2-1) [left  =-6pt of q2-1, xshift=-0pt, font=\scriptsize] {$(1,1)$};
	    	\node[state] (q3-1) [below  =1cm of q1, xshift=2cm] {$v_3$};
	    	\node[state] (lq3-1) [right  =-6pt of q3-1, xshift=-0pt, font=\scriptsize] {$(1,1)$};
	    	\node[state] (q2-2) [below  =1cm of q1, xshift=-0.8cm] {$v_2$};
	    	\node[state] (lq2-2) [left  =-8pt of q2-2, xshift=-0pt, font=\scriptsize] {$(1,1)$};
	    	\node[state] (q3-2) [below  =1cm of q1, xshift=0.8cm] {$v_3$};
	    	\node[state] (lq3-2) [right  =-8pt of q3-2, xshift=-0pt, font=\scriptsize] {$(1,1)$};
	    	\node[state] (q4) [below  =1cm of q2-1, xshift=-1cm] {$v_4$};
	    	\node[state] (lq4) [right  =-6pt of q4, xshift=-0pt, font=\scriptsize] {$(1,1)$};
	    	\node[state] (q5) [below  =1cm of q3-1, xshift=-2cm] {$v_5$};
	    	\node[state] (lq5) [left  =-6pt of q5, xshift=-0pt, font=\scriptsize] {$(1,1)$};
	    	\node[state] (q6) [below  =1cm of q3-1, xshift=1cm] {$v_6$};
	    	\node[state] (lq6) [left  =-6pt of q6, xshift=-0pt, font=\scriptsize] {$(1,1)$};
	    	\node[state] (q7) [below  =1cm of q5, xshift=-1.5cm] {$v_7$};
	    	\node[state] (lq7) [left  =-6pt of q7, xshift=-0pt, font=\scriptsize] {$(1,1)$};
	    	\node[state] (q8) [below  =1cm of q5, xshift=1.5cm] {$v_8$};
	    	\node[state] (lq8) [right  =-6pt of q8, xshift=-0pt, font=\scriptsize] {$(1,1)$};
	    	\node[state] (q9) [below  =1cm of q8, xshift=-1.5cm] {$v_9$};
	    	
	    	\path[->] 
	    	(q1) edge node[below=-4pt, xshift=5pt, font=\scriptsize]{$e_1$} (q2-1)
	    	     edge node[below=-4pt, xshift=-5pt, font=\scriptsize]{$e_2$} (q3-1)
	    	     edge node[below=-4pt, xshift=5pt, font=\scriptsize]{$e_1$} (q2-2)
	    	     edge node[below=-4pt, xshift=-5pt, font=\scriptsize]{$e_2$} (q3-2)
	    	     
	    	(q2-1) edge node[below=-4pt, xshift=5pt, font=\scriptsize]{$e_3$} (q4)
	    	(q2-2) edge node[below=-4pt, xshift=-5pt, font=\scriptsize]{$e_4$} (q5)
	    	(q3-2) edge node[below=-4pt, xshift=5pt, font=\scriptsize]{$e_5$} (q5)
	    	(q3-1) edge node[below=-4pt, xshift=-5pt, font=\scriptsize]{$e_6$} (q6)
	    	(q4) edge node[below=-4pt, xshift=-5pt, font=\scriptsize]{$e_7$} (q7)
	    	(q5) edge node[below=-4pt, xshift=5pt, font=\scriptsize]{$e_8$} (q7)
	    	     edge node[below=-4pt, xshift=-5pt, font=\scriptsize]{$e_9$} (q8)
	    	(q6) edge node[below=-4pt, xshift=5pt, font=\scriptsize]{$e_{10}$} (q8)
	    	(q7) edge node[below=-4pt, xshift=-5pt, font=\scriptsize]{$e_{11}$} (q9)
	    	(q8) edge node[below=-4pt, xshift=5pt, font=\scriptsize]{$e_{12}$} (q9)
	    	;
	    	\end{tikzpicture}
	    	\caption{}
	    \end{subfigure}
        }
        \resizebox{0.4\textwidth}{!}{
        \begin{subfigure}[b!]{0.65\textwidth}
        	\centering
        	\begin{tikzpicture}[>=stealth,thick]
        	\usetikzlibrary{calc}
        	\tikzstyle{state}=[]
        	
        	\node[state] (q1) [xshift=0cm] {$v_9$};
        	\node[state] (q2) [above  =0.5cm of q1, xshift=-2.4cm] {$v_7$};
        	\node[state] (q3) [above  =0.5cm of q1, xshift=2.4cm] {$v_8$};
        	\node[state] (q4) [above  =1cm of q2, xshift=-1cm] {$v_4$};
        	
        	\node[state] (q5-1) [above  =0.5cm of q2, xshift=1cm] {$v_5$};
        	\node[state] (q7-1) [above  =1.5cm of q5-1, xshift=-2cm] {$v_2$};
        	\node[state] (q7-3) [above  =1cm of q5-1, xshift=-0.5cm] {$v_2$};
        	\node[state] (q8-1) [above  =1cm of q5-1, xshift=0.8cm] {$v_3$};
        	\node[state] (q9) [above  =1.5cm of q8-1, xshift=0.7cm] {$v_1$};
        	
        	\node[state] (q6) [above  =1cm of q3, xshift=1cm] {$v_6$};
        	\node[state] (q5-2) [above  =0.5cm of q3, xshift=-1cm] {$v_5$};
        	\node[state] (q7-2) [above  =1cm of q5-2, xshift=-0.8cm] {$v_2$};
        	\node[state] (q8-2) [above  =1.5cm of q5-2, xshift=2cm] {$v_3$};
        	\node[state] (q8-3) [above  =1cm of q5-2, xshift=0.5cm] {$v_3$};
        	
        	\path[<-] 
        	(q1) edge node[below=-4pt, xshift=-5pt, font=\scriptsize]{$e_{11}$} (q2)
        	     edge node[below=-4pt, xshift=5pt, font=\scriptsize]{$e_{12}$} (q3)
        	(q2) edge node[below=-4pt, xshift=-5pt, font=\scriptsize]{$e_{7}$} (q4)
        	     edge node[below=-4pt, xshift=5pt, font=\scriptsize]{$e_{8}$} (q5-1)
        	
        	(q3) edge node[below=-4pt, xshift=5pt, font=\scriptsize]{$e_{10}$} (q6)
        	     edge node[below=-4pt, xshift=-5pt, font=\scriptsize]{$e_{9}$} (q5-2)
        	(q4) edge node[below=-4pt, xshift=-5pt, font=\scriptsize]{$e_{3}$} (q7-1)
        	(q5-1) edge node[above=-0pt, xshift=3pt, font=\scriptsize]{$e_{4}$} (q7-3)
        	       edge node[below=-4pt, xshift=5pt, font=\scriptsize]{$e_{5}$} (q8-1)
        	(q7-1) edge node[below=-4pt, xshift=5pt, font=\scriptsize]{$e_{1}$} (q9)
        	(q7-3) edge node[below=-0pt, xshift=2pt, font=\scriptsize]{$e_{1}$} (q9)
        	(q8-1) edge node[below=-15pt, xshift=-1pt, font=\scriptsize]{$e_{2}$} (q9)
        	
        	(q6) edge node[below=-4pt, xshift=5pt, font=\scriptsize]{$e_{6}$} (q8-2)
        	(q5-2) edge node[below=-4pt, xshift=-5pt, font=\scriptsize]{$e_{4}$} (q7-2)
        	       edge node[above=-0pt, xshift=-3pt, font=\scriptsize]{$e_{5}$} (q8-3)
        	(q7-2) edge node[above=1pt, xshift=2pt, font=\scriptsize]{$e_{1}$} (q9)
        	(q8-2) edge node[below=-4pt, xshift=-5pt, font=\scriptsize]{$e_{2}$} (q9)
        	(q8-3) edge node[below=-0pt, xshift=-2pt, font=\scriptsize]{$e_{2}$} (q9)
        	;
        	\end{tikzpicture}

        	\caption{}
        \end{subfigure}
        }
	\caption{Complex Block Factorization}\label{fig:2-step-factorization}
\end{figure}

There are two methods to factorize a complex block: one is forward factorization, 
the other is backward factorization. 

To apply backward factorization, we first compute the in-degree of 
every intermediate vertex in the block, treating every simple block 
whose intermediate vertices do not include the vertex being computed as an edge. 
e.g., in \cref*{fig:complex-dag-and-complete-path-dag} (a), 
the in-degree of vertex $v_4$ is $1$. 
Then, from bottom to top, we search for vertices whose in-degrees are greater than 1, 
split from behind the outgoing edges, 
reducing the overlapping degrees of outgoing edges to 1. 
e.g., in \cref*{fig:complex-dag-and-complete-path-dag} (b), 
$v_7$ and $v_8$ are the predecessors of sink $v_9$, their in-degrees are both 2, 
i.e., $O_{P_{v_7}}(e_{11})=2$ and $O_{P_{v_8}}(e_{12})=2$, 
which means there are 2 paths overlapped on $e_{11}$ and $e_{12}$ respectively, 
we split from behind $e_{11}$ and $e_{12}$, 
by creating another duplicated $v_7$ and $v_8$, $e_{11}$ and $e_{12}$ in locations 
that are different from the original one, 
connecting $e_8$ and $e_9$ to the new $v_7$ and new $v_8$ respectively, 
as shown in \cref*{fig:2-step-factorization} (a). 
After that, we adjust the in-degrees of vertices and continue to search for next target vertex 
from among the predecessor vertices of current position, 
e.g., the in-degree of $v_5$ in \cref*{fig:2-step-factorization} (a) is 2. 
Since $e\langle v_5,v_9\rangle$ is a simple block, we regard it as one edge, 
split it just like how we split $e_{11}$ and $e_{12}$ in the previous step, 
as shown in \cref*{fig:2-step-factorization} (b). 

Forward factorization is basically the same as backward factorization 
except in the opposite direction.
Instead of tracking in-degrees of intermediate vertices, 
searching for target vertices from bottom up and splitting from behind an edge, 
it tracks out-degrees, 
searches for target vertices from top down 
and splits from the head of an edge. 
\cref*{fig:2-step-factorization} (c) and (d) illustrate the process of forward factorization.

We can quickly convert a simple block to an algebraic expression.
The algebraic expression of \cref*{fig:2-step-factorization} (b) is:

\begin{equation}\label{eq-b-simple}
\frac{\mathrm{d} v_1}{\mathrm{d} v_7}=
e_1(e_3e_7e_{11}+e_4(e_8e_{11}+e_9e_{12}))+
e_2(e_6e_{10}e_{12}+e_5(e_8e_{11}+e_9e_{12}))
\end{equation}
\\
The algebraic expression of \cref*{fig:2-step-factorization} (d) is:

\begin{equation}\label{eq-d-simple}
\frac{\mathrm{d} v_1}{\mathrm{d} v_7}=
(e_1e_3e_7+(e_1e_4+e_2e_5)e_8)e_{11}+
(e_2e_6e_{10}+(e_1e_4+e_2e_5)e_9)e_{12}
\end{equation}
\\
Note that in \cref*{eq-b-simple} and \cref*{eq-d-simple},
the same intermediate variable or subexpression $(e_8e_{11}+e_9e_{12})$ and $(e_1e_4+e_2e_5)$
cannot be factored out despite they can be merged in the graph. 
In practice, we would like to use some intermediate variables to offset 
the limitation of a single algebraic expression.
e.g., let $s_1=e_8e_{11}+e_9e_{12}$, \cref*{eq-b-simple} becomes: 

\begin{equation}\label{eq-b-simple-t-var}
\frac{\mathrm{d} v_1}{\mathrm{d} v_7}=
e_1(e_3e_7e_{11}+e_4s_1)+
e_2(e_6e_{10}e_{12}+e_5s_1)
\end{equation}

If we put $(e_8e_{11}+e_9e_{12})$ under the \cref*{eq-b-simple-t-var}, 
and draw two directed edges from the bottom of $s_1$'s in different locations
to the top of it, denoting `$=$' or reference, it becomes a DAG. 
We notice that multiple algebraic expressions connected by intermediate variables is 
effectively a DAG of references. 
During the process of complex block factorization, 
in order to take advantage of the reference structure of multiple algebraic expressions, 
for every simple block that splits like an edge, 
we can create a new reference edge representing it, 
and replace it with that edge before splitting. 
e.g., in \cref*{fig:factorization-with-reference}, 
a new edge $s_1$ is created to replace $e\langle v_5,v_9\rangle$.

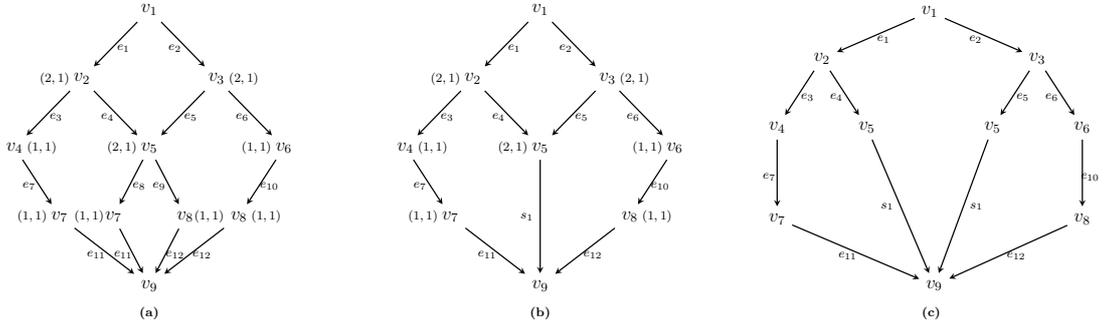
\begin{figure}[H] 
	\centering
	\resizebox{0.3\textwidth}{!}{
		\begin{subfigure}[b]{0.5\textwidth}
			\centering
			\begin{tikzpicture}[>=stealth,thick]
			\usetikzlibrary{calc}
			\tikzstyle{state}=[]
			
			\node[state] (q1) [xshift=0cm] {$v_1$};
			\node[state] (q2) [below  =1cm of q1, xshift=-1.5cm] {$v_2$};
			\node[state] (lq2) [left  =-6pt of q2, xshift=-0pt, font=\scriptsize] {$(2,1)$};
			\node[state] (q3) [below  =1cm of q1, xshift=1.5cm] {$v_3$};
			\node[state] (lq3) [right  =-6pt of q3, xshift=-0pt, font=\scriptsize] {$(2,1)$};
			\node[state] (q4) [below  =1cm of q2, xshift=-1.5cm] {$v_4$};
			\node[state] (lq4) [right  =-6pt of q4, xshift=-0pt, font=\scriptsize] {$(1,1)$};
			\node[state] (q5) [below  =1cm of q3, xshift=-1.5cm] {$v_5$};
			\node[state] (lq5) [left  =-6pt of q5, xshift=-0pt, font=\scriptsize] {$(2,1)$};
			\node[state] (q6) [below  =1cm of q3, xshift=1.5cm] {$v_6$};
			\node[state] (lq6) [left  =-6pt of q6, xshift=-0pt, font=\scriptsize] {$(1,1)$};
			\node[state] (q7-1) [below  =1cm of q5, xshift=-2.0cm] {$v_7$};
			\node[state] (lq7-1) [left  =-6pt of q7-1, xshift=-0pt, font=\scriptsize] {$(1,1)$};
			\node[state] (q7-2) [below  =1cm of q5, xshift=-.8cm] {$v_7$};
			\node[state] (lq7-2) [left  =-8pt of q7-2, xshift=-0pt, font=\scriptsize] {$(1,1)$};
			\node[state] (q8-1) [below  =1cm of q5, xshift=2.0cm] {$v_8$};
			\node[state] (lq8-1) [right  =-6pt of q8-1, xshift=-0pt, font=\scriptsize] {$(1,1)$};
			\node[state] (q8-2) [below  =1cm of q5, xshift=.8cm] {$v_8$};
			\node[state] (lq8-2) [right  =-8pt of q8-2, xshift=-0pt, font=\scriptsize] {$(1,1)$};
			\node[state] (q9) [below  =1cm of q8-1, xshift=-2.0cm] {$v_9$};
			
			\path[->] 
			(q1) edge node[below=-4pt, xshift=5pt, font=\scriptsize]{$e_1$} (q2)
			edge node[below=-4pt, xshift=-5pt, font=\scriptsize]{$e_2$} (q3)
			(q2) edge node[below=-4pt, xshift=5pt, font=\scriptsize]{$e_3$} (q4)
			edge node[below=-4pt, xshift=-5pt, font=\scriptsize]{$e_4$} (q5)
			(q3) edge node[below=-4pt, xshift=5pt, font=\scriptsize]{$e_5$} (q5)
			edge node[below=-4pt, xshift=-5pt, font=\scriptsize]{$e_6$} (q6)
			(q4) edge node[below=-4pt, xshift=-5pt, font=\scriptsize]{$e_7$} (q7-1)
			(q5) edge node[below=-4pt, xshift=5pt, font=\scriptsize]{$e_8$} (q7-2)
			edge node[below=-4pt, xshift=-5pt, font=\scriptsize]{$e_9$} (q8-2)
			(q6) edge node[below=-4pt, xshift=5pt, font=\scriptsize]{$e_{10}$} (q8-1)
			(q7-1) edge node[below=-4pt, xshift=-5pt, font=\scriptsize]{$e_{11}$} (q9)
			(q8-1) edge node[below=-4pt, xshift=5pt, font=\scriptsize]{$e_{12}$} (q9)
			(q7-2) edge node[below=-4pt, xshift=-5pt, font=\scriptsize]{$e_{11}$} (q9)
			(q8-2) edge node[below=-4pt, xshift=5pt, font=\scriptsize]{$e_{12}$} (q9)
			;
			\end{tikzpicture}
			
			\caption{}
		\end{subfigure}
	}
	\resizebox{0.3\textwidth}{!}{
	\begin{subfigure}[b]{0.5\textwidth}
		\centering
		\begin{tikzpicture}[>=stealth,thick]
		\usetikzlibrary{calc}
		\tikzstyle{state}=[]
		\node[state] (q1) [xshift=0cm] {$v_1$};
		\node[state] (q2) [below  =1cm of q1, xshift=-1.5cm] {$v_2$};
		\node[state] (lq2) [left  =-6pt of q2, xshift=-0pt, font=\scriptsize] {$(2,1)$};
		\node[state] (q3) [below  =1cm of q1, xshift=1.5cm] {$v_3$};
		\node[state] (lq3) [right  =-6pt of q3, xshift=-0pt, font=\scriptsize] {$(2,1)$};
		\node[state] (q4) [below  =1cm of q2, xshift=-1.5cm] {$v_4$};
		\node[state] (lq4) [right  =-6pt of q4, xshift=-0pt, font=\scriptsize] {$(1,1)$};
		\node[state] (q5) [below  =1cm of q3, xshift=-1.5cm] {$v_5$};
		\node[state] (lq5) [left  =-6pt of q5, xshift=-0pt, font=\scriptsize] {$(2,1)$};
		\node[state] (q6) [below  =1cm of q3, xshift=1.5cm] {$v_6$};
		\node[state] (lq6) [left  =-6pt of q6, xshift=-0pt, font=\scriptsize] {$(1,1)$};
		\node[state] (q7-1) [below  =1cm of q5, xshift=-2.0cm] {$v_7$};
		\node[state] (lq7-1) [left  =-6pt of q7-1, xshift=-0pt, font=\scriptsize] {$(1,1)$};

		\node[state] (q8-1) [below  =1cm of q5, xshift=2.0cm] {$v_8$};
		\node[state] (lq8-1) [right  =-6pt of q8-1, xshift=-0pt, font=\scriptsize] {$(1,1)$};

		\node[state] (q9) [below  =1cm of q8-1, xshift=-2.0cm] {$v_9$};
		
		\path[->] 
		(q1) edge node[below=-4pt, xshift=5pt, font=\scriptsize]{$e_1$} (q2)
		     edge node[below=-4pt, xshift=-5pt, font=\scriptsize]{$e_2$} (q3)
		(q2) edge node[below=-4pt, xshift=5pt, font=\scriptsize]{$e_3$} (q4)
		     edge node[below=-4pt, xshift=-5pt, font=\scriptsize]{$e_4$} (q5)
		(q3) edge node[below=-4pt, xshift=5pt, font=\scriptsize]{$e_5$} (q5)
		     edge node[below=-4pt, xshift=-5pt, font=\scriptsize]{$e_6$} (q6)
		(q4) edge node[below=-4pt, xshift=-5pt, font=\scriptsize]{$e_7$} (q7-1)
		(q5) edge node[left=-0pt, xshift=0pt, font=\scriptsize]{$s_1$} (q9)
		(q6) edge node[below=-4pt, xshift=5pt, font=\scriptsize]{$e_{10}$} (q8-1)
		(q7-1) edge node[below=-4pt, xshift=-5pt, font=\scriptsize]{$e_{11}$} (q9)
		(q8-1) edge node[below=-4pt, xshift=5pt, font=\scriptsize]{$e_{12}$} (q9)

		;
		\end{tikzpicture}
		
		\caption{}
	\end{subfigure}
    }
	\resizebox{0.3\textwidth}{!}{
		\begin{subfigure}[b]{0.5\textwidth}
			\centering
			\begin{tikzpicture}[>=stealth,thick]
			\usetikzlibrary{calc}
			\tikzstyle{state}=[]
			
			\node[state] (q1) [xshift=0cm] {$v_1$};
			\node[state] (q2) [below  =0.5cm of q1, xshift=-2.4cm] {$v_2$};
			\node[state] (q3) [below  =0.5cm of q1, xshift=2.4cm] {$v_3$};
			\node[state] (q4) [below  =1cm of q2, xshift=-1cm] {$v_4$};
			
			\node[state] (q5-1) [below  =1cm of q2, xshift=1cm] {$v_5$};
			\node[state] (q5-2) [below  =1cm of q3, xshift=-1cm] {$v_5$};
			\node[state] (q7-1) [below  =1.5cm of q5-1, xshift=-2cm] {$v_7$};
			\node[state] (q8-2) [below  =1.5cm of q5-2, xshift=2cm] {$v_8$};

			\node[state] (q9) [below  =3cm of q5-1, xshift=1.5cm] {$v_9$};
			
			\node[state] (q6) [below  =1cm of q3, xshift=1cm] {$v_6$};

			\path[->] 
			(q1) edge node[below=-4pt, xshift=5pt, font=\scriptsize]{$e_1$} (q2)
			     edge node[below=-4pt, xshift=-5pt, font=\scriptsize]{$e_2$} (q3)
			(q2) edge node[below=-4pt, xshift=5pt, font=\scriptsize]{$e_3$} (q4)
		   	     edge node[below=-4pt, xshift=-5pt, font=\scriptsize]{$e_4$} (q5-1)
			(q3) edge node[below=-4pt, xshift=-5pt, font=\scriptsize]{$e_6$} (q6)
			     edge node[below=-4pt, xshift=5pt, font=\scriptsize]{$e_5$} (q5-2)
			(q4) edge node[below=-4pt, xshift=-5pt, font=\scriptsize]{$e_7$} (q7-1)
			(q5-1) edge node[left=-0pt, xshift=0pt, font=\scriptsize]{$s_1$} (q9)
			(q7-1) edge node[below=-4pt, xshift=-5pt, font=\scriptsize]{$e_{11}$} (q9)
			(q6) edge node[below=-4pt, xshift=5pt, font=\scriptsize]{$e_{10}$} (q8-2)
			(q5-2) edge node[right=0pt, xshift=0pt, font=\scriptsize]{$s_1$} (q9)
			(q8-2) edge node[below=-4pt, xshift=5pt, font=\scriptsize]{$e_{12}$} (q9)

			;
			\end{tikzpicture}
			\caption{}
		\end{subfigure}
	}

	\caption{Factorization with Reference}\label{fig:factorization-with-reference}
\end{figure}

Differentiation graph is path-oriented, edges and vertices are the carrier of paths, 
multiple paths could overlap on the same edges and vertices. 
In a direct simple chain, all the edges and vertices form only one path, 
therefore elimination of path is equivalent to elimination of corresponding edges and vertices.
We notice that simple chains and simple blocks could be nested structures 
like algebraic expressions. 
In an optimized arithmetic expression, the optimal calculation order is from inside out, 
so are the simple blocks and simple chains. 
When the most inner direct simple chains and direct simple blocks are eliminated 
and replaced by shortcut edges, 
the second most inner indirect simple chains and indirect simple blocks 
become direct simple chains and direct simple blocks. 
We can repeat this all the way to the outermost simple chains and simple blocks.
Thus, for the simple blocks and simple chains, 
elimination of paths is equivalent to elimination of corresponding edges and vertices
as long as it follows the optimal calculation order. 
Notice that in \cref*{fig:factorization-with-reference} (a), 
$e\langle v_5,v_9\rangle$ and $(e_8e_{11}+e_9e_{12})$ are equivalent, 
$(e_8e_{11}+e_9e_{12})$ is the algebraic form, $e\langle v_5,v_9\rangle$ is the graph form. 
We can easily convert a simple block or a simple chain to its corresponding algebraic form, 
and then deduce it back again from the algebraic expression we have, 
they are interconvertible. 
It is worth noting that the commutativity of multiplication relation in algebraic expression, 
e.g., $e_ie_j$ and $e_je_i$, makes sense 
while in differentiation graph makes no sense 
because $e_i\succ e_j$ or $e_i^{-}\succ e_i^{+}=e_j^{-}\succ e_j^{+}$. 
Therefore, the precondition of their interconvertibility is 
to impose noncommutativity restriction of multiplication relation 
on the algebraic expressions converted from differentiation graph. 

\begin{equation}\label{eq-jaconbian-accumulation}
\underbrace{
\left[
e_1, 
e_2
\right]
}_A
\times
\underbrace{
\left[
\begin{matrix}
e_3 & e_4 & 0 \\
0 & e_5 & e_6
\end{matrix}
\right]
}_B
\times
\underbrace{
\left[
\begin{matrix}
e_7 & 0 \\
e_8 & e_9 \\
0 & e_{10}
\end{matrix}
\right]
}_C
\times
\underbrace{
\left[
\begin{matrix}
e_{11} \\
e_{12} 
\end{matrix}
\right]
}_D
\end{equation} 

Notice that we can convert \cref*{fig:complex-dag-and-complete-path-dag} (b) 
to the form of local Jacobians, as shown in \cref*{eq-jaconbian-accumulation}. 
If we use $(AB)$ to denote the matrix product of $A$ and $B$, 
then $(AB)CD$ is the \cref*{fig:2-step-factorization} (c), 
$((AB)C)D$ is the \cref*{fig:2-step-factorization} (d), 
$AB(CD)$ is the \cref*{fig:2-step-factorization} (a), 
$A(B(CD))$ is the \cref*{fig:2-step-factorization} (b). 
The accumulation process of $ABCD$ is 
the factorization process of differentiation graph 
in \cref*{fig:complex-dag-and-complete-path-dag} (b). 
If we accumulate it from left to right, i.e. $(((AB)C)D)$, the result is \cref*{eq-d-simple}; 
from right to left, i.e. $(A(B(CD)))$, the result is \cref*{eq-b-simple}. 
Similar to the reference edge in \cref*{fig:factorization-with-reference}, 
we can also use reference variables to replace the entries in $(AB)$. 
It is worth noting that there are more than just two accumulation orders we discussed above, 
we can accumulate local Jacobians from both sides or from the middle, 
e.g., $((AB)(CD))$ or $((A(BC))D)$, 
which means we can also factorize complex blocks from both sides or from the middle. 
Different accumulation orders result in different algebraic expressions with different structures, 
our goal is to find an algebraic expression that requires less fma's to compute. 
Note that local Jacobian accumulation merges or eliminates depth levels 
through matrix multiplications.  
Despite it can automatically factorize complex blocks 
by treating every entry of the result local Jacobian 
as an independent simple block or simple chain 
and packing their accumulations into one atomic operation, 
it ignores the optimal calculation order in the simple chains and simple blocks, 
which means it is up to us to find a better accumulation order. 
e.g., in \cref*{fig:complex-dag-and-complete-path-dag} (a), 
we prefer $(J^0_1J^1_2)(J^2_3J^3_4)$ to $J^0_1(J^1_2J^2_3)J^3_4$, 
because the latter splits $v_4$ into 4 vertices, 
resulting in \cref*{fig:complex-dag-and-complete-path-dag} (c), 
which requires more fma's to compute. 

\begin{figure}[H] 
	\centering
	\resizebox{.25\textwidth}{!}{
		\begin{subfigure}[b]{0.38\textwidth}
			\centering
			\begin{tikzpicture}[>=stealth,thick]
			\usetikzlibrary{calc}
			\tikzstyle{state}=[]
			
			\node[state] (q1) [below  =0.7cm of q0, xshift=0cm] {$v_1$};
			\node[state] (q2) [below  =0.5cm of q1, xshift=-1cm] {$v_2$};
			\node[state] (q3) [below  =0.5cm of q1, xshift=1cm] {$v_3$};
			\node[state] (q4) [below  =0.5cm of q3, xshift=-1cm] {$v_4$};
			\node[state] (q5) [below  =0.5cm of q4, xshift=-1cm] {$v_5$};
			\node[state] (q6) [below  =0.5cm of q4, xshift=1cm] {$v_6$};
			\node[state] (q7) [below  =0.5cm of q6, xshift=-1cm] {$v_7$};
			\node[state] (q8) [below  =0.5cm of q7, xshift=0cm] {$v_8$};
			
			\node[state] (q10) [below  =0.7cm of q0, xshift=3cm] {$v_{10}$};
			\node[state] (q11) [below  =0.5cm of q10, xshift=0cm] {$v_{11}$};
			\node[state] (q12) [below  =0.5cm of q11, xshift=-1cm] {$v_{12}$};
			\node[state] (q13) [below  =0.5cm of q11, xshift=1cm] {$v_{13}$};
			\node[state] (q14) [below  =0.5cm of q13, xshift=-1cm] {$v_{14}$};
			\node[state] (q15) [below  =0.5cm of q14, xshift=-1cm] {$v_{15}$};
			\node[state] (q16) [below  =0.5cm of q14, xshift=1cm] {$v_{16}$};
			\node[state] (q17) [below  =0.5cm of q16, xshift=-1cm] {$v_{17}$};
			
			\path[->] 
			
			(q1) edge node[below=-4pt, xshift=5pt, font=\scriptsize]{$e_1$} (q2)
			     edge node[below=-4pt, xshift=-5pt, font=\scriptsize]{$e_2$} (q3)
			(q2) edge node[below=-4pt, xshift=-5pt, font=\scriptsize]{$e_3$} (q4)
			(q3) edge node[below=-4pt, xshift=5pt, font=\scriptsize]{$e_4$} (q4)
			(q4) edge node[below=-4pt, xshift=5pt, font=\scriptsize]{$e_5$} (q5)
			     edge node[below=-4pt, xshift=-5pt, font=\scriptsize]{$e_6$} (q6)
			(q5) edge node[below=-4pt, xshift=-5pt, font=\scriptsize]{$e_7$} (q7)
			(q6) edge node[below=-4pt, xshift=5pt, font=\scriptsize]{$e_8$} (q7)
			(q7) edge node[below=-7pt, xshift=5pt, font=\scriptsize]{$e_9$} (q8)
			
			(q10) edge node[below=-7pt, xshift=7pt, font=\scriptsize]{$e_{10}$} (q11)
			(q11) edge node[below=-4pt, xshift=5pt, font=\scriptsize]{$e_{11}$} (q12)
			      edge node[below=-4pt, xshift=-5pt, font=\scriptsize]{$e_{12}$} (q13)
			(q12) edge node[below=-4pt, xshift=-5pt, font=\scriptsize]{$e_{13}$} (q14)
			(q13) edge node[below=-4pt, xshift=5pt, font=\scriptsize]{$e_{14}$} (q14)
			(q14) edge node[below=-4pt, xshift=5pt, font=\scriptsize]{$e_{15}$} (q15)
			      edge node[below=-4pt, xshift=-5pt, font=\scriptsize]{$e_{16}$} (q16)
			(q15) edge node[below=-4pt, xshift=-5pt, font=\scriptsize]{$e_{17}$} (q17)
			(q16) edge node[below=-4pt, xshift=5pt, font=\scriptsize]{$e_{18}$} (q17)
			;
			\end{tikzpicture}
			\caption{}
       \end{subfigure}
	}
    \resizebox{.25\textwidth}{!}{
    	\begin{subfigure}[b]{0.38\textwidth}
    		\centering
    		\begin{tikzpicture}[>=stealth,thick]
    		\usetikzlibrary{calc}
    		\tikzstyle{state}=[]
    		
    		\node[state] (q1) [below  =0.7cm of q0, xshift=0cm] {$v_1$};
    		\node[state] (q2) [below  =0.5cm of q1, xshift=-1cm] {};
    		\node[state] (q3) [below  =0.5cm of q1, xshift=1cm] {};
    		\node[state] (q4) [below  =1.6cm of q1, xshift=0cm] {$v_4$};
    		\node[state] (q5) [below  =0.5cm of q4, xshift=-1cm] {$v_5$};
    		\node[state] (q6) [below  =0.5cm of q4, xshift=1cm] {$v_6$};
    		\node[state] (q7) [below  =0.5cm of q6, xshift=-1cm] {$v_7$};
    		\node[state] (q8) [below  =0.5cm of q7, xshift=0cm] {$v_8$};
    		
    		\node[state] (q10) [below  =0.7cm of q0, xshift=3cm] {$v_{10}$};
    		\node[state] (q11) [below  =0.5cm of q10, xshift=0cm] {$v_{11}$};
    		\node[state] (q12) [below  =0.5cm of q11, xshift=-1cm] {};
    		\node[state] (q13) [below  =0.5cm of q11, xshift=1cm] {};
    		\node[state] (q14) [below  =1.6cm of q11, xshift=0cm] {$v_{14}$};
    		\node[state] (q15) [below  =0.5cm of q14, xshift=-1cm] {$v_{15}$};
    		\node[state] (q16) [below  =0.5cm of q14, xshift=1cm] {$v_{16}$};
    		\node[state] (q17) [below  =0.5cm of q16, xshift=-1cm] {$v_{17}$};
    		
    		\path[->] 
    		
            (q1) edge node[below=-4pt, xshift=6pt, font=\scriptsize]{$s_1$} (q4)
    		(q4) edge node[below=-4pt, xshift=5pt, font=\scriptsize]{$e_5$} (q5)
    		     edge node[below=-4pt, xshift=-5pt, font=\scriptsize]{$e_6$} (q6)
    		(q5) edge node[below=-4pt, xshift=-5pt, font=\scriptsize]{$e_7$} (q7)
    		(q6) edge node[below=-4pt, xshift=5pt, font=\scriptsize]{$e_8$} (q7)
    		(q7) edge node[below=-7pt, xshift=5pt, font=\scriptsize]{$e_9$} (q8)
    		
    		(q10) edge node[below=-7pt, xshift=7pt, font=\scriptsize]{$e_{10}$} (q11)
            (q11) edge node[below=-4pt, xshift=6pt, font=\scriptsize]{$s_{2}$} (q14)
    		(q14) edge node[below=-4pt, xshift=5pt, font=\scriptsize]{$e_{15}$} (q15)
    		      edge node[below=-4pt, xshift=-5pt, font=\scriptsize]{$e_{16}$} (q16)
    		(q15) edge node[below=-4pt, xshift=-5pt, font=\scriptsize]{$e_{17}$} (q17)
    		(q16) edge node[below=-4pt, xshift=5pt, font=\scriptsize]{$e_{18}$} (q17)
    		;
    		\end{tikzpicture}
    		\caption{}
    	\end{subfigure}
    }

	\caption{Conflicting depth accumulation order}
	\label{fig:conflicting-depth-accumulation-order}
\end{figure}
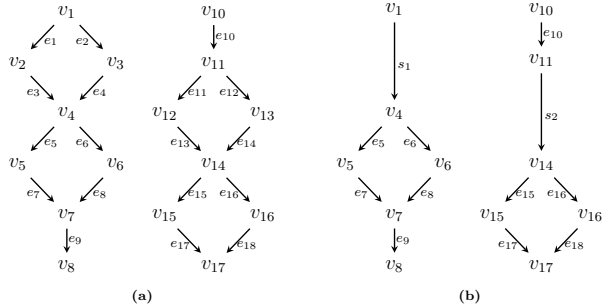

When accumulating local Jacobians of different depth levels, 
sometime different subgraphs may have conflicting optimal accumulation orders, 
e.g., in \cref*{fig:conflicting-depth-accumulation-order} (a), 
the optimal accumulation order for the left subgraph is $(J^0_1J^1_2)(J^2_3J^3_4)J^4_5$, 
while for the right subgraph is $J^0_1(J^1_2J^2_3)(J^3_4J^4_5)$. 
Although depth local Jacobian multiplication can be performed separately 
in unrelated sub-graphs that span the same depth levels, 
local Jacobians are not limited to two different depth levels, 
we can also apply non-depth local Jacobian accumulation, e.g., the result of 
$J\langle v_1,v_{11}|v_2,v_3,v_{12},v_{13}\rangle\times 
J\langle v_2,v_3,v_{12},v_{13}|v_4,v_{14}\rangle=
J\langle v_1,v_{11}|v_4,v_{14}\rangle$
is \cref*{fig:conflicting-depth-accumulation-order} (b), 
where $s_1=e_1e_3+e_2e_4$ and $s_2=e_{11}e_{13}+e_{12}e_{14}$. 
In fact, 
we can first eliminate all the simple chains and simple blocks 
in a differentiation graph using reference edge, 
and eliminate depth levels 
by performing depth local Jacobian multiplication 
when accumulating layered dense graph or subgraph.

\begin{figure}[H] 
	\centering
	\resizebox{0.3\textwidth}{!}{
		\begin{subfigure}[b]{0.45\textwidth}
			\centering
			\begin{tikzpicture}[>=stealth,thick]
			\usetikzlibrary{calc}
			\tikzstyle{state}=[]
			
			\node[state] (y1) {$y_1$};
			\node[state] (e1) [below  =0.7cm of y1, xshift=-2cm] {$e_1$};
			\node[state] (e2) [below  =0.7cm of y1, xshift=2cm] {$e_2$};
			
			\node[state] (e3) [below  =0.7cm of e1, xshift=-1cm] {$e_3$};
			\node[state] (e4) [below  =0.7cm of e1, xshift=1cm] {$e_4$};
			\node[state] (e5) [below  =0.7cm of e2, xshift=-1cm] {$e_5$};
			\node[state] (e6) [below  =0.7cm of e2, xshift=1cm] {$e_6$};
			
			\node[state] (e7) [below  =0.7cm of e3, xshift=0cm] {$e_7$};
            \node[state] (e8) [below  =0.7cm of e4, xshift=0cm] {$e_8$};
            \node[state] (e9) [below  =0.7cm of e5, xshift=0cm] {$e_9$};
            \node[state] (e10) [below  =0.7cm of e6, xshift=0cm] {$e_{10}$};
            
            \node[state] (e11) [below  =0.7cm of e7, xshift=1cm] {$e_{11}$};
            \node[state] (e12) [below  =0.7cm of e10, xshift=-1cm] {$e_{12}$};
            
            \node[state] (x9) [below  =0.7cm of e11, xshift=2cm] {$x_{9}$};

			\path[->] 
			(y1) edge (e1)
			     edge (e2)
			(e1) edge (e3)
			     edge (e4)
			(e2) edge (e5)
			     edge (e6)
			(e4) edge (e8)
			     edge (e9)
			(e5) edge (e8)
			     edge (e9)
			(e3) edge (e7)
			(e6) edge (e10)
			(e8) edge (e11)
            (e9) edge (e12)
			(e7) edge (e11)
            (e10) edge (e12)
			(e11) edge (x9)
            (e12) edge (x9)
			;
            \end{tikzpicture}
            \caption{}
        \end{subfigure}
    }
	\resizebox{0.3\textwidth}{!}{
	\begin{subfigure}[b]{0.45\textwidth}
		\centering
		\begin{tikzpicture}[>=stealth,thick]
		\usetikzlibrary{calc}
		\tikzstyle{state}=[]
		
		\node[state] (y1) {$y_1$};
		\node[state] (e1) [below  =0.7cm of y1, xshift=-2cm] {$e_1$};
		\node[state] (e2) [below  =0.7cm of y1, xshift=2cm] {$e_2$};
		
		\node[state] (e3) [below  =0.7cm of e1, xshift=-1cm] {$e_3$};
		\node[state] (e4) [below  =0.7cm of e1, xshift=1cm] {$e_4$};
		\node[state] (e5) [below  =0.7cm of e2, xshift=-1cm] {$e_5$};
		\node[state] (e6) [below  =0.7cm of e2, xshift=1cm] {$e_6$};
		
		\node[state] (e7) [below  =0.7cm of e3, xshift=0cm] {$e_7$};
		\node[state] (e8) [below  =0.7cm of e4, xshift=0cm] {$e_8$};
		\node[state] (e9) [below  =0.7cm of e5, xshift=0cm] {$e_9$};
		\node[state] (e10) [below  =0.7cm of e6, xshift=0cm] {$e_{10}$};
		
		\node[state] (e11) [below  =0.7cm of e7, xshift=1cm] {$e_{11}$};
		\node[state] (e11-1) [below  =0.7cm of e8, xshift=0.25cm] {$e_{11}$};
		\node[state] (e12) [below  =0.7cm of e10, xshift=-1cm] {$e_{12}$};
		\node[state] (e12-1) [below  =0.7cm of e9, xshift=-0.25cm] {$e_{12}$};
		
		\node[state] (x9) [below  =0.7cm of e11, xshift=2cm] {$x_{9}$};
		
		\path[->] 
		(y1) edge (e1)
		     edge (e2)
		(e1) edge (e3)
		     edge (e4)
		(e2) edge (e5)
		     edge (e6)
		(e4) edge (e8)
		     edge (e9)
		(e5) edge (e8)
		     edge (e9)
		(e3) edge (e7)
        (e6) edge (e10)
		(e8) edge (e11-1)
		(e9) edge (e12-1)
		(e7) edge (e11)
		(e10) edge (e12)
		(e11) edge (x9)
		(e12) edge (x9)
		(e11-1) edge (x9)
        (e12-1) edge (x9)
		
		;
		\end{tikzpicture}
		\caption{}
	\end{subfigure}
    }
    \resizebox{0.3\textwidth}{!}{
    	\begin{subfigure}[b]{0.45\textwidth}
    		\centering
    		\begin{tikzpicture}[>=stealth,thick]
    		\usetikzlibrary{calc}
    		\tikzstyle{state}=[]
    		
    		\node[state] (y1) {$y_1$};
    		\node[state] (e1) [below  =0.7cm of y1, xshift=-2cm] {$e_1$};
    		\node[state] (e2) [below  =0.7cm of y1, xshift=2cm] {$e_2$};
    		
    		\node[state] (e3) [below  =0.7cm of e1, xshift=-1cm] {$e_3$};
    		\node[state] (e4) [below  =0.7cm of e1, xshift=1cm] {$e_4$};
    		\node[state] (e5) [below  =0.7cm of e2, xshift=-1cm] {$e_5$};
    		\node[state] (e6) [below  =0.7cm of e2, xshift=1cm] {$e_6$};
    		
    		\node[state] (s1) [below  =0.7cm of e4, xshift=1cm] {$s_1$};
    		
    		\node[state] (e7) [below  =0.7cm of e3, xshift=0cm] {$e_7$};
    		\node[state] (e10) [below  =0.7cm of e6, xshift=0cm] {$e_{10}$};
    		
    		\node[state] (e11) [below  =0.7cm of e7, xshift=1cm] {$e_{11}$};
    		\node[state] (e12) [below  =0.7cm of e10, xshift=-1cm] {$e_{12}$};
    		
    		\node[state] (x9) [below  =0.7cm of e11, xshift=2cm] {$x_{9}$};
    		
    		\path[->] 
    		(y1) edge (e1)
    		edge (e2)
    		(e1) edge (e3)
    		edge (e4)
    		(e2) edge (e5)
    		edge (e6)

    		(e4) edge (s1)
    		(e5) edge (s1)
    		(e3) edge (e7)
    		(e6) edge (e10)
    		(e7) edge (e11)
    		(e10) edge (e12)
    		(e11) edge (x9)
    		(e12) edge (x9)
     		(s1) edge (x9)
    		
    		;
    		\end{tikzpicture}
    		\caption{}
    	\end{subfigure}
    }
	\resizebox{0.3\textwidth}{!}{
	\begin{subfigure}[b]{0.45\textwidth}
		\centering
		\begin{tikzpicture}[>=stealth,thick]
		\usetikzlibrary{calc}
		\tikzstyle{state}=[]
		
		\node[state] (y1) {$y_1$};
		\node[state] (e1) [below  =0.7cm of y1, xshift=-2cm] {$e_1$};
		\node[state] (e2) [below  =0.7cm of y1, xshift=2cm] {$e_2$};
		
		\node[state] (e3) [below  =0.7cm of e1, xshift=-1cm] {}; 
		\node[state] (e4) [below  =0.7cm of e1, xshift=1cm] {$e_4$};
		\node[state] (e5) [below  =0.7cm of e2, xshift=-1cm] {$e_5$};
		\node[state] (e6) [below  =0.7cm of e2, xshift=1cm] {};
		
		\node[state] (e3e7) [below  =0.7cm of e3, xshift=0cm] {$e_3e_7$};
		\node[state] (e8) [below  =0.7cm of e4, xshift=0cm] {$e_8$};
		\node[state] (e9) [below  =0.7cm of e5, xshift=0cm] {$e_9$};
		\node[state] (e6e10) [below  =0.7cm of e6, xshift=0cm] {$e_6e_{10}$};
		
		\node[state] (e11) [below  =0.7cm of e7, xshift=1cm] {$e_{11}$};
		\node[state] (e12) [below  =0.7cm of e10, xshift=-1cm] {$e_{12}$};
		
		\node[state] (x9) [below  =0.7cm of e11, xshift=2cm] {$x_{9}$};

		\path[->] 
		(y1) edge (e1)
		     edge (e2)
		(e1) edge (e3e7)
		     edge (e4)
		(e2) edge (e5)
		     edge (e6e10)
		(e4) edge (e8)
		     edge (e9)
		(e5) edge (e8)
		     edge (e9)
		(e8) edge (e11)
		(e9) edge (e12)
		(e3e7) edge (e11)
		(e6e10) edge (e12)
		(e11) edge (x9)
		(e12) edge (x9)
		
		;
		\end{tikzpicture}
		\caption{}
	\end{subfigure}
    }
	\resizebox{0.3\textwidth}{!}{
	\begin{subfigure}[b]{0.45\textwidth}
		\centering
		\begin{tikzpicture}[>=stealth,thick]
		\usetikzlibrary{calc}
		\tikzstyle{state}=[]
		
		\node[state] (y1) {$y_1$}; 
		\node[state] (e1) [below  =0.7cm of y1, xshift=-2cm] {$e_1$};
		\node[state] (e2) [below  =0.7cm of y1, xshift=2cm] {$e_2$};
		
		\node[state] (e3) [below  =0.7cm of e1, xshift=-1cm] {}; 
		\node[state] (e4) [below  =0.7cm of e1, xshift=1cm] {$e_4$};
		\node[state] (e5) [below  =0.7cm of e2, xshift=-1cm] {$e_5$};
		\node[state] (e6) [below  =0.7cm of e2, xshift=1cm] {}; 
		
		\node[state] (e3e7e11) [below  =0.7cm of e3, xshift=0cm] {$e_3e_7e_{11}$};
		\node[state] (e8) [below  =0.7cm of e4, xshift=0cm] {}; 
		\node[state] (e9) [below  =0.7cm of e5, xshift=0cm] {$e_9$};
		\node[state] (e6e10) [below  =0.7cm of e6, xshift=0cm] {$e_6e_{10}$};
		
		\node[state] (e8e11) [below  =0.7cm of e8, xshift=0.5cm] {$e_8e_{11}$};
		\node[state] (e12) [below  =0.7cm of e10, xshift=-1cm] {$e_{12}$};
		
		\node[state] (x9) [below  =0.7cm of e11, xshift=2cm] {$x_{9}$};
		
		\path[->] 
		(y1) edge (e1)
		     edge (e2)
		(e1) edge (e3e7)
	      	 edge (e4)
		(e2) edge (e5)
		     edge (e6e10)
		(e4) edge (e8e11)
		     edge (e9)
		(e5) edge (e8e11)
		     edge (e9)
		(e9) edge (e12)
		(e6e10) edge (e12)
		(e3e7e11) edge (x9)
		(e12) edge (x9)
		(e8e11) edge (x9)
		
		;
		\end{tikzpicture}
		\caption{}
	\end{subfigure}
    }
	\resizebox{0.3\textwidth}{!}{
	\begin{subfigure}[b]{0.45\textwidth}
		\centering
		\begin{tikzpicture}[>=stealth,thick]
		\usetikzlibrary{calc}
		\tikzstyle{state}=[]
		
		\node[state] (y1) {$y_1$}; 
		\node[state] (e1) [below  =0.7cm of y1, xshift=-2cm] {$e_1$};
		\node[state] (e2) [below  =0.7cm of y1, xshift=2cm] {$e_2$};
		
		\node[state] (e3) [below  =0.7cm of e1, xshift=-1cm] {};
		\node[state] (e4) [below  =0.7cm of e1, xshift=1cm] {$e_4$};
		\node[state] (e5) [below  =0.7cm of e2, xshift=-1cm] {$e_5$};
		\node[state] (e6) [below  =0.7cm of e2, xshift=1cm] {};
		
		\node[state] (e3e7e11) [below  =0.7cm of e3, xshift=0cm] {$e_3e_7e_{11}$};
		\node[state] (e8) [below  =0.7cm of e4, xshift=0cm] {}; 
		\node[state] (e6e10e12) [below  =0.7cm of e6, xshift=0cm] {$e_6e_{10}e_{12}$};
		
		\node[state] (e11) [below  =0.7cm of e7, xshift=1cm] {};
		\node[state] (s1) [below  =1.0cm of e4, xshift=1.0cm] {$s_1$};
		\node[state] (e12) [below  =0.7cm of e10, xshift=-1cm] {};
		
		\node[state] (x9) [below  =0.7cm of e11, xshift=2cm] {$x_{9}$};
		
		\path[->] 
		(y1) edge (e1)
		edge (e2)
		(e1) edge (e3e7e11)
		     edge (e4)
		(e2) edge (e5)
		     edge (e6e10e12)
		(e4) edge (s1)
		(e5) edge (s1)

		(e3e7e11) edge (x9)
		(e6e10e12) edge (x9)
		(s1) edge (x9)
		;
		\end{tikzpicture}
		\caption{}
	\end{subfigure}
    }
    \caption{Chain-Block Elimination in the Line-Graph}\label{fig:chain-block-elimination}
    \end{figure}
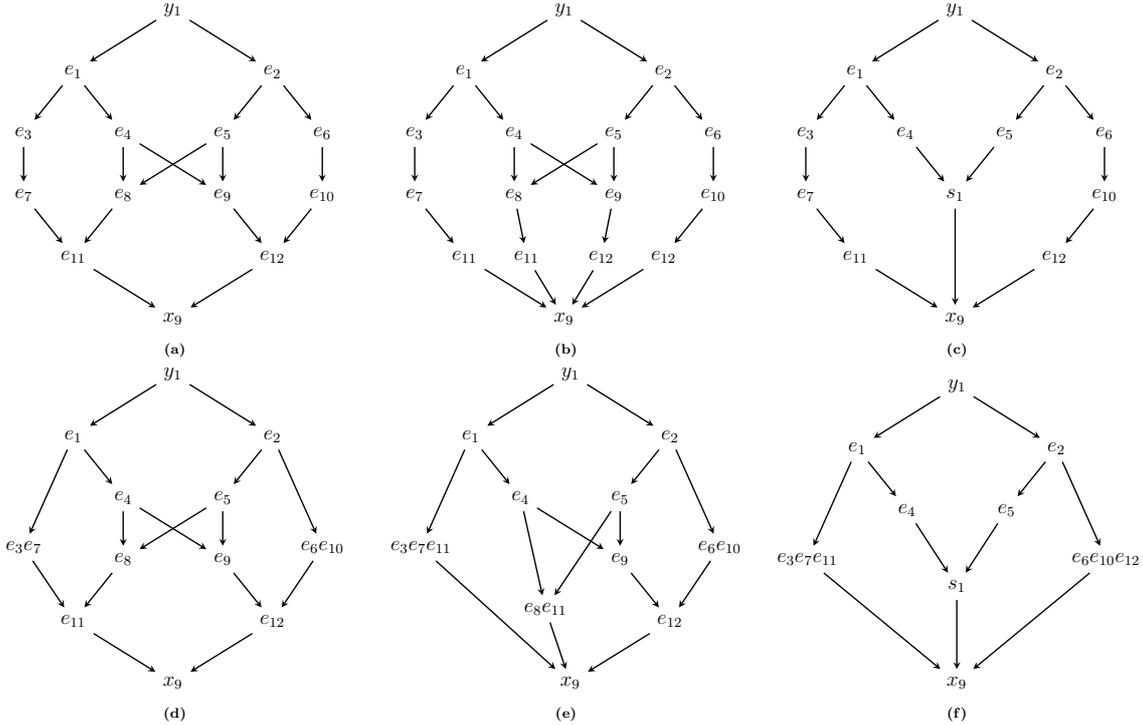

Note that we can perform face elimination in the line-graph 
in the same way how we factorize the differentiation graph, 
splitting vertices explicitly, then eliminating edges, 
as shown in \cref*{fig:chain-block-elimination} (a) (b) (c). 
We can also perform face elimination using the rules presented by Naumann, 
splitting vertices implicitly, 
as shown in \cref*{fig:chain-block-elimination} (d) (e) (f). 
The first way is more similar to the corresponding algebraic expression, 
while the second way costs less when implementing the process in a program, 
because it creates less new vertices. 
\\

\section{Multi-Root or Multi-Terminal}

So far we have only considered the situation of $(|Y|=1)\cap (|X|=1)$. 
When $|Y|>1$ or $|X|>1$,  
multiple blocks or chains are intertwined together, 
the subgraph of every root-terminal pair is a chain or block. 
Then, how do we factorize a graph with multiple root vertices or terminal vertices?
As Griewank and Walther state in their book that 
``...a large number of additive terms that have many common multiplicative subexpressions.
Hence, the challenge is now to organize the computation such that 
a large number of these subexpressions is identified 
and reused not only for a single pair $(j,i)$ but across all nonzero partials $a_{ij}$.''
(\cite{griewank2008} 2008, p196)
We need to identify the common sub-chains or sub-blocks 
nested in root-terminal chains or blocks, 
and replace them with reference edges. 

The purpose of multiple blocks or chains being intertwined together is 
to merge the same common sub-chains or sub-blocks. 
If the outmost blocks or chains have no nested common sub-chains or sub-blocks, 
there is no need to put them together to compute, 
they can be easily separated from each other. 
Just like local Jacobian accumulation 
sees every entry of the result local Jacobian independently, 
another solution is to 
simply list or completely split all the sub-graphs of root-terminal pairs 
from the differentiation graph 
as the entries of root-terminal Jacobian, 
since they do not share any common sub-chain or sub-block that needs be computed first. 
e.g., we can list all the root-terminal pairs in \cref*{fig:biclique} (a). 
In other words, if we know the optimal common sub-chain or sub-block combination set, 
we can compute them separately. 
The upper bound is that we ignore the common sub-chains or sub-blocks, 
separating root-terminal pairs directly and optimizing them individually. 

We notice that different combination set of root-terminal chains or blocks may have 
different set of common sub-chains or sub-blocks.
Let $A=\{(v_i,v_j)|v_i\in Y,v_j\in X\}$, 
$B=\{S|S\subseteq A\}$, $C=\{S|S\subseteq B\}$, 
then $|B|=2^{|A|}=2^{|X|\cdot|Y|}$, $|C|=2^{|B|}=2^{2^{|X|\cdot|Y|}}$. 
When $|X|$ and $|Y|$ increase, 
$|B|$ increases exponentially and $|C|$ increases doubly exponentially, 
it becomes extremely difficult to conduct a brute-force search for a optimal combination 
that requires minimum number of fma's to compute the root-terminal Jacobian. 
However, the combinations in $B$ are not equally important, 
we need to narrow down the searching space, 
e.g., find the largest common sub-chains or sub-blocks overlapped by most root-terminal pairs, 
where `largest' means largest number of fma's needed to compute it. 
Another problem is that when a solution is found, how do we prove that it is optimal? 
What if there is another solution that requires less fma's to calculate?
It is hard to verify if the solution is optimal, especially when $|X|\cdot|Y|$ is large. 
What we can do is to find a solution that requires as less fma's as possible. 

Note that the optimal dynamic value accumulation in linearized differentiation graph 
is not the optimal calculation order. 
If the cost of searching for an optimal calculation order 
exceeds the calculation cost saved by the optimal calculation order, 
we would rather not do it. 
However, for a static symbolic differentiation graph that calculates a lot of different inputs, 
the optimal calculation order matters.

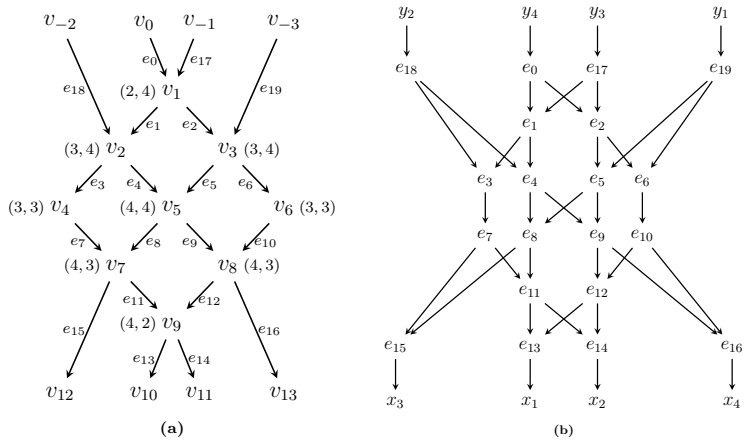
\begin{figure}[H] 
	\centering
	\resizebox{0.3\textwidth}{!}{
	\begin{subfigure}[b]{0.4\textwidth}
		\centering
		\begin{tikzpicture}[>=stealth,thick]
		\usetikzlibrary{calc}
		\tikzstyle{state}=[]
		
		\node[state] (q0) [xshift=0cm] {$v_0$};
		\node[state] (q-1) [xshift=1cm] {$v_{-1}$};
		\node[state] (q-2) [xshift=-1.5cm] {$v_{-2}$};
		\node[state] (q-3) [xshift=2.5cm] {$v_{-3}$};
		
		\node[state] (q1) [below  =0.7cm of q0, xshift=0.5cm] {$v_1$};
		\node[state] (lq1) [left  =-6pt of q1, xshift=-0pt, font=\scriptsize] {$(2,4)$};
		\node[state] (q2) [below  =0.5cm of q1, xshift=-1cm] {$v_2$};
		\node[state] (lq2) [left  =-6pt of q2, xshift=-0pt, font=\scriptsize] {$(3,4)$};
		\node[state] (q3) [below  =0.5cm of q1, xshift=1cm] {$v_3$};
		\node[state] (lq3) [right  =-6pt of q3, xshift=-0pt, font=\scriptsize] {$(3,4)$};
		\node[state] (q4) [below  =0.5cm of q2, xshift=-1cm] {$v_4$};
		\node[state] (lq4) [left  =-6pt of q4, xshift=-0pt, font=\scriptsize] {$(3,3)$};
		\node[state] (q5) [below  =0.5cm of q3, xshift=-1cm] {$v_5$};
		\node[state] (lq5) [left  =-6pt of q5, xshift=-0pt, font=\scriptsize] {$(4,4)$};
		\node[state] (q6) [below  =0.5cm of q3, xshift=1cm] {$v_6$};
		\node[state] (lq6) [right  =-6pt of q6, xshift=-0pt, font=\scriptsize] {$(3,3)$};
		\node[state] (q7) [below  =0.5cm of q5, xshift=-1cm] {$v_7$};
		\node[state] (lq7) [left  =-6pt of q7, xshift=-0pt, font=\scriptsize] {$(4,3)$};
		\node[state] (q8) [below  =0.5cm of q5, xshift=1cm] {$v_8$};
		\node[state] (lq8) [right  =-6pt of q8, xshift=-0pt, font=\scriptsize] {$(4,3)$};
		\node[state] (q9) [below  =0.5cm of q8, xshift=-1cm] {$v_9$};
		\node[state] (lq9) [left  =-6pt of q9, xshift=-0pt, font=\scriptsize] {$(4,2)$};
		
		\node[state] (q10) [below  =0.7cm of q9, xshift=-0.5cm] {$v_{10}$};
		\node[state] (q11) [below  =0.7cm of q9, xshift=0.5cm] {$v_{11}$};
		\node[state] (q12) [below  =0.7cm of q9, xshift=-2cm] {$v_{12}$};
		\node[state] (q13) [below  =0.7cm of q9, xshift=2cm] {$v_{13}$};
		
		\path[->] 
		(q0) edge node[below=-5pt, xshift=-4pt, font=\scriptsize]{$e_0$} (q1)
		(q-1) edge node[below=-5pt, xshift=8pt, font=\scriptsize]{$e_{17}$} (q1)
		(q-2) edge node[below=-5pt, xshift=-7pt, font=\scriptsize]{$e_{18}$} (q2)
		(q-3) edge node[below=-5pt, xshift=8pt, font=\scriptsize]{$e_{19}$} (q3)
		
		(q1) edge node[below=-4pt, xshift=5pt, font=\scriptsize]{$e_1$} (q2)
		edge node[below=-4pt, xshift=-5pt, font=\scriptsize]{$e_2$} (q3)
		(q2) edge node[below=-4pt, xshift=5pt, font=\scriptsize]{$e_3$} (q4)
		edge node[below=-4pt, xshift=-5pt, font=\scriptsize]{$e_4$} (q5)
		(q3) edge node[below=-4pt, xshift=5pt, font=\scriptsize]{$e_5$} (q5)
		edge node[below=-4pt, xshift=-5pt, font=\scriptsize]{$e_6$} (q6)
		(q4) edge node[below=-4pt, xshift=-5pt, font=\scriptsize]{$e_7$} (q7)
		(q5) edge node[below=-4pt, xshift=5pt, font=\scriptsize]{$e_8$} (q7)
		edge node[below=-4pt, xshift=-5pt, font=\scriptsize]{$e_9$} (q8)
		(q6) edge node[below=-4pt, xshift=5pt, font=\scriptsize]{$e_{10}$} (q8)
		(q7) edge node[below=-4pt, xshift=-5pt, font=\scriptsize]{$e_{11}$} (q9)
		(q8) edge node[below=-4pt, xshift=5pt, font=\scriptsize]{$e_{12}$} (q9)
		
		(q9) edge node[below=-6pt, xshift=-7pt, font=\scriptsize]{$e_{13}$} (q10)
		edge node[below=-6pt, xshift=7pt, font=\scriptsize]{$e_{14}$} (q11)
		(q7) edge node[below=-6pt, xshift=-7pt, font=\scriptsize]{$e_{15}$} (q12)
		(q8) edge node[below=-6pt, xshift=7pt, font=\scriptsize]{$e_{16}$} (q13)
		;
		\end{tikzpicture}
		\caption{}
	\end{subfigure}
    }
	\resizebox{.3\textwidth}{!}{
		\begin{subfigure}[b]{0.5\textwidth}
			\centering
			\begin{tikzpicture}[>=stealth,thick]
			\usetikzlibrary{calc}
			\tikzstyle{state}=[]
			
			\node[state] (y2) [xshift=0cm] {$y_2$};
			\node[state] (y4) [xshift=2.75cm] {$y_4$};
			\node[state] (y3) [xshift=4.25cm] {$y_3$};
			\node[state] (y1) [xshift=7cm] {$y_1$};
			
			\node[state] (e18) [below  =0.7cm of y2, xshift=-0cm] {$e_{18}$};
			\node[state] (e0) [below  =0.7cm of y4, xshift=-0cm] {$e_{0}$};
			\node[state] (e17) [below  =0.7cm of y3, xshift=0cm] {$e_{17}$};
			\node[state] (e19) [below  =0.7cm of y1, xshift=0cm] {$e_{19}$};
			
			\node[state] (e1) [below  =0.7cm of e0, xshift=-0cm] {$e_{1}$};
			\node[state] (e2) [below  =0.7cm of e17, xshift=0cm] {$e_{2}$};
			
			\node[state] (e3) [below  =0.7cm of e1, xshift=-1cm] {$e_{3}$};
			\node[state] (e4) [below  =0.7cm of e1, xshift=-0cm] {$e_{4}$};
			\node[state] (e5) [below  =0.7cm of e2, xshift=0cm] {$e_{5}$};
			\node[state] (e6) [below  =0.7cm of e2, xshift=1cm] {$e_{6}$};
			
			\node[state] (e7) [below  =0.7cm of e3, xshift=-0cm] {$e_{7}$};
			\node[state] (e8) [below  =0.7cm of e4, xshift=-0cm] {$e_{8}$};
			\node[state] (e9) [below  =0.7cm of e5, xshift=0cm] {$e_{9}$};
			\node[state] (e10) [below  =0.7cm of e6, xshift=0cm] {$e_{10}$};
			
			\node[state] (e11) [below  =0.7cm of e8, xshift=0.0cm] {$e_{11}$};
			\node[state] (e12) [below  =0.7cm of e9, xshift=-0.0cm] {$e_{12}$};
			
			\node[state] (e15) [below  =0.7cm of e11, xshift=-3cm] {$e_{15}$};
			\node[state] (e13) [below  =0.7cm of e11, xshift=-0cm] {$e_{13}$};
			\node[state] (e14) [below  =0.7cm of e12, xshift=0cm] {$e_{14}$};
			\node[state] (e16) [below  =0.7cm of e12, xshift=3cm] {$e_{16}$};
			
			\node[state] (x3) [below  =0.7cm of e15, xshift=-0cm] {$x_{3}$};
			\node[state] (x1) [below  =0.7cm of e13, xshift=-0cm] {$x_{1}$};
			\node[state] (x2) [below  =0.7cm of e14, xshift=0cm] {$x_{2}$};
			\node[state] (x4) [below  =0.7cm of e16, xshift=0cm] {$x_{4}$};
			
			\path[->] 
			(y2) edge (e18)
			(y4) edge (e0)
			(y3) edge (e17)
			(y1) edge (e19)
			
			(e18) edge (e3)
			edge (e4)
			(e0) edge (e1)
			edge (e2)
			(e17) edge (e1)
			edge (e2)
			(e19) edge (e5)
			edge (e6)
			
			(e1) edge (e3)
			edge (e4)
			(e2) edge (e5)
			edge (e6)
			
			(e3) edge (e7)
			(e6) edge (e10)
			(e4) edge (e8)
			edge (e9)
			(e5) edge (e8)
			edge (e9)
			
			(e7) edge (e11)
			edge (e15)
			(e8) edge (e11)
			edge (e15)
			(e9) edge (e12)
			edge (e16)
			(e10) edge (e12)
			edge (e16)
			
			(e11) edge (e13)
			      edge (e14)
			(e12) edge (e13)
			      edge (e14)
			
			(e15) edge (x3)
			(e13) edge (x1)
			(e14) edge (x2)
			(e16) edge (x4)
			
			;
			\end{tikzpicture}
			\caption{}
		\end{subfigure}
	}
	\caption{Multi-root and Multi-terminal Differentiation Graph and Line Graph}
	\label{fig:multi-root-terminal-graph}
\end{figure}

For example, \cref*{fig:multi-root-terminal-graph} (a) is 
a multi-root and multi-terminal differentiation graph. 
We notice that if we ignore all the vertices and edges not in $P_{v_{-2}\to v_{12}}$, 
i.e., all the vertices and edges that are irrelevant to the calculation of 
$\frac{\partial v_{-2}}{\partial v_{12}}$, 
then all the relevant vertices and edges in $P_{v_{-2}\to v_{12}}$ 
form an indirect simple chain $e'\langle v_{-2},v_{12}\rangle$. 
We can quickly write down its algebraic form: $e_{18}(e_3e_7+e_4e_8)e_{15}$. 
However after finishing the calculation or conversion, we find ourselves in an awkward situation:   
it cannot be eliminated directly, 
because $\forall e_{k}\in P_{v_{-2}\to v_{12}}, \mathrm{O}_{P_{Y\to X}}(e_{k})>1$. 
Similarly, in its corresponding line-graph \cref*{fig:multi-root-terminal-graph} (b), 
it cannot be eliminated directly either, 
because $\forall e_{k}\in P_{y_{2}\to x_{3}}, \mathrm{O}_{P_{\tilde{Y}\to \tilde{X}}}(e_{k})>1$.

\begin{figure}[H] 
	\centering
	\resizebox{0.3\textwidth}{!}{
	\begin{subfigure}[b]{0.4\textwidth}
		\centering
		\begin{tikzpicture}[>=stealth,thick]
		\usetikzlibrary{calc}
		\tikzstyle{state}=[]
		
		\node[state] (q0) [xshift=0cm] {$v_0$};
		\node[state] (q-1) [xshift=1cm] {$v_{-1}$};
		\node[state] (q-2) [xshift=-1.5cm] {$v_{-2}$};
		\node[state] (q-3) [xshift=2.5cm] {$v_{-3}$};
		\node[state] (q-4) [xshift=4cm] {$v_{-4}$};
		\node[state] (q14) [below  =2cm of q-4, xshift=0cm] {$v_{14}$};
		\node[state] (lq14) [left  =-6pt of q14, xshift=-0pt, font=\scriptsize] {$(1,1)$};
		
		\node[state] (q1) [below  =0.7cm of q0, xshift=0.5cm] {$v_1$};
		\node[state] (lq1) [left  =-6pt of q1, xshift=-0pt, font=\scriptsize] {$(3,4)$};
		\node[state] (q2) [below  =0.5cm of q1, xshift=-1cm] {$v_2$};
		\node[state] (lq2) [left  =-6pt of q2, xshift=-0pt, font=\scriptsize] {$(3,4)$};
		\node[state] (q3) [below  =0.5cm of q1, xshift=1cm] {$v_3$};
		\node[state] (lq3) [right  =-6pt of q3, xshift=-0pt, font=\scriptsize] {$(3,4)$};
		\node[state] (q4) [below  =0.5cm of q2, xshift=-1cm] {$v_4$};
		\node[state] (lq4) [left  =-6pt of q4, xshift=-0pt, font=\scriptsize] {$(3,3)$};
		\node[state] (q5) [below  =0.5cm of q3, xshift=-1cm] {$v_5$};
		\node[state] (lq5) [left  =-6pt of q5, xshift=-0pt, font=\scriptsize] {$(4,4)$};
		\node[state] (q6) [below  =0.5cm of q3, xshift=1cm] {$v_6$};
		\node[state] (lq6) [right  =-6pt of q6, xshift=-0pt, font=\scriptsize] {$(3,3)$};
		\node[state] (q7) [below  =0.5cm of q5, xshift=-1cm] {$v_7$};
		\node[state] (lq7) [left  =-6pt of q7, xshift=-0pt, font=\scriptsize] {$(4,3)$};
		\node[state] (q8) [below  =0.5cm of q5, xshift=1cm] {$v_8$};
		\node[state] (lq8) [right  =-6pt of q8, xshift=-0pt, font=\scriptsize] {$(4,3)$};
		\node[state] (q9) [below  =0.5cm of q8, xshift=-1cm] {$v_9$};
		\node[state] (lq9) [left  =-6pt of q9, xshift=-0pt, font=\scriptsize] {$(4,2)$};
		
		\node[state] (q10) [below  =0.7cm of q9, xshift=-0.5cm] {$v_{10}$};
		\node[state] (q11) [below  =0.7cm of q9, xshift=0.5cm] {$v_{11}$};
		\node[state] (q12) [below  =0.7cm of q9, xshift=-2cm] {$v_{12}$};
		\node[state] (q13) [below  =0.7cm of q9, xshift=2cm] {$v_{13}$};
		
		\path[->] 
		(q0) edge node[below=-5pt, xshift=-4pt, font=\scriptsize]{$e_0$} (q1)
		(q-1) edge node[below=-5pt, xshift=8pt, font=\scriptsize]{$e_{17}$} (q1)
		(q-2) edge node[below=-5pt, xshift=-7pt, font=\scriptsize]{$e_{18}$} (q2)
		      edge node[below=-5pt, xshift=-7pt, font=\scriptsize]{$e_{20}$} (q1)
		(q-3) edge node[below=-5pt, xshift=8pt, font=\scriptsize]{$e_{19}$} (q3)
		(q-4) edge node[below=-5pt, xshift=-7pt, font=\scriptsize]{$e_{21}$} (q14)
		
		(q1) edge node[below=-4pt, xshift=5pt, font=\scriptsize]{$e_1$} (q2)
		edge node[below=-4pt, xshift=-5pt, font=\scriptsize]{$e_2$} (q3)
		(q2) edge node[below=-4pt, xshift=5pt, font=\scriptsize]{$e_3$} (q4)
		     edge node[below=-4pt, xshift=-5pt, font=\scriptsize]{$e_4$} (q5)
		(q3) edge node[below=-4pt, xshift=5pt, font=\scriptsize]{$e_5$} (q5)
		edge node[below=-4pt, xshift=-5pt, font=\scriptsize]{$e_6$} (q6)
		(q4) edge node[below=-4pt, xshift=-5pt, font=\scriptsize]{$e_7$} (q7)
		(q5) edge node[below=-4pt, xshift=5pt, font=\scriptsize]{$e_8$} (q7)
		edge node[below=-4pt, xshift=-5pt, font=\scriptsize]{$e_9$} (q8)
		(q6) edge node[below=-4pt, xshift=5pt, font=\scriptsize]{$e_{10}$} (q8)
		(q7) edge node[below=-4pt, xshift=-5pt, font=\scriptsize]{$e_{11}$} (q9)
		(q8) edge node[below=-4pt, xshift=5pt, font=\scriptsize]{$e_{12}$} (q9)
		
		(q9) edge node[below=-6pt, xshift=-7pt, font=\scriptsize]{$e_{13}$} (q10)
		edge node[below=-6pt, xshift=7pt, font=\scriptsize]{$e_{14}$} (q11)
		(q7) edge node[below=-6pt, xshift=-7pt, font=\scriptsize]{$e_{15}$} (q12)
		(q8) edge node[below=-6pt, xshift=7pt, font=\scriptsize]{$e_{16}$} (q13)
		(q14) edge node[below=-5pt, xshift=8pt, font=\scriptsize]{$e_{22}$} (q13)
		;
		\end{tikzpicture}
		\caption{}
	\end{subfigure}
    }
	\resizebox{0.3\textwidth}{!}{
	\begin{subfigure}[b]{0.4\textwidth}
		\centering
		\begin{tikzpicture}[>=stealth,thick]
		\usetikzlibrary{calc}
		\tikzstyle{state}=[]
		
		\node[state] (q0) [xshift=0cm] {$v_0$};
		\node[state] (q-1) [xshift=1cm] {$v_{-1}$};
		\node[state] (q-2) [xshift=-1.5cm] {$v_{-2}$};
		\node[state] (q-3) [xshift=2.5cm] {$v_{-3}$};
		\node[state] (q-4) [xshift=4cm] {$v_{-4}$};
		
		\node[state] (q1) [below  =0.7cm of q0, xshift=0.5cm] {$v_1$};
		\node[state] (lq1) [left  =-6pt of q1, xshift=-0pt, font=\scriptsize] {$(3,4)$};
		\node[state] (q2) [below  =0.5cm of q1, xshift=-1cm] {$v_2$};
		\node[state] (lq2) [left  =-6pt of q2, xshift=-0pt, font=\scriptsize] {$(3,4)$};
		\node[state] (q3) [below  =0.5cm of q1, xshift=1cm] {$v_3$};
		\node[state] (lq3) [right  =-6pt of q3, xshift=-0pt, font=\scriptsize] {$(3,4)$};
		\node[state] (q5) [below  =0.5cm of q3, xshift=-1cm] {$v_5$};
		\node[state] (lq5) [left  =-9pt of q5, xshift=-0pt, font=\scriptsize] {$(4,4)$};
		\node[state] (q7) [below  =0.5cm of q5, xshift=-1cm] {$v_7$};
		\node[state] (lq7) [left  =-6pt of q7, xshift=-0pt, font=\scriptsize] {$(4,3)$};
		\node[state] (q8) [below  =0.5cm of q5, xshift=1cm] {$v_8$};
		\node[state] (lq8) [right  =-6pt of q8, xshift=-0pt, font=\scriptsize] {$(4,3)$};
		\node[state] (q9) [below  =0.5cm of q8, xshift=-1cm] {$v_9$};
		\node[state] (lq9) [left  =-6pt of q9, xshift=-0pt, font=\scriptsize] {$(4,2)$};
		
		\node[state] (q10) [below  =0.7cm of q9, xshift=-0.5cm] {$v_{10}$};
		\node[state] (q11) [below  =0.7cm of q9, xshift=0.5cm] {$v_{11}$};
		\node[state] (q12) [below  =0.7cm of q9, xshift=-2cm] {$v_{12}$};
		\node[state] (q13) [below  =0.7cm of q9, xshift=2cm] {$v_{13}$};
		
		\path[->] 
		(q0) edge node[below=-5pt, xshift=-4pt, font=\scriptsize]{$e_0$} (q1)
		(q-1) edge node[below=-5pt, xshift=8pt, font=\scriptsize]{$e_{17}$} (q1)
		(q-2) edge node[below=-5pt, xshift=-7pt, font=\scriptsize]{$e_{18}$} (q2)
		      edge node[below=-5pt, xshift=-7pt, font=\scriptsize]{$e_{20}$} (q1)
		(q-3) edge node[below=-5pt, xshift=8pt, font=\scriptsize]{$e_{19}$} (q3)
		(q-4) edge node[below=-5pt, xshift=-7pt, font=\scriptsize]{$s_{3}$} (q13)
		
		(q1) edge node[below=-4pt, xshift=5pt, font=\scriptsize]{$e_1$} (q2)
		     edge node[below=-4pt, xshift=-5pt, font=\scriptsize]{$e_2$} (q3)
		(q2) edge node[left=2pt, xshift=5pt, font=\scriptsize]{$s_1$} (q7)
		     edge node[below=-4pt, xshift=-5pt, font=\scriptsize]{$e_4$} (q5)
		(q3) edge node[below=-4pt, xshift=5pt, font=\scriptsize]{$e_5$} (q5)
		     edge node[right=2pt, xshift=-5pt, font=\scriptsize]{$s_2$} (q8)
		(q5) edge node[below=-4pt, xshift=5pt, font=\scriptsize]{$e_8$} (q7)
		     edge node[below=-4pt, xshift=-5pt, font=\scriptsize]{$e_9$} (q8)
		(q7) edge node[below=-4pt, xshift=-5pt, font=\scriptsize]{$e_{11}$} (q9)
		(q8) edge node[below=-4pt, xshift=5pt, font=\scriptsize]{$e_{12}$} (q9)
		
		(q9) edge node[below=-6pt, xshift=-7pt, font=\scriptsize]{$e_{13}$} (q10)
		     edge node[below=-6pt, xshift=7pt, font=\scriptsize]{$e_{14}$} (q11)
		(q7) edge node[below=-6pt, xshift=-7pt, font=\scriptsize]{$e_{15}$} (q12)
		(q8) edge node[below=-6pt, xshift=7pt, font=\scriptsize]{$e_{16}$} (q13)
		;
		\end{tikzpicture}
		\caption{}
	\end{subfigure}
    }
	\resizebox{0.3\textwidth}{!}{
	\begin{subfigure}[b]{0.4\textwidth}
		\centering
		\begin{tikzpicture}[>=stealth,thick]
		\usetikzlibrary{calc}
		\tikzstyle{state}=[]
		
		\node[state] (q0) [xshift=0cm] {$v_0$};
		\node[state] (q-1) [xshift=1cm] {$v_{-1}$};
		\node[state] (q-2) [xshift=-1.5cm] {$v_{-2}$};
		\node[state] (q-3) [xshift=2.5cm] {$v_{-3}$};
		
		\node[state] (q1) [below  =0.7cm of q0, xshift=0.5cm] {$v_1$};
		\node[state] (lq1) [left  =-6pt of q1, xshift=-0pt, font=\scriptsize] {$(3,4)$};
		\node[state] (q2) [below  =0.5cm of q1, xshift=-1cm] {$v_2$};
		\node[state] (lq2) [left  =-6pt of q2, xshift=-0pt, font=\scriptsize] {$(3,4)$};
		\node[state] (q3) [below  =0.5cm of q1, xshift=1cm] {$v_3$};
		\node[state] (lq3) [right  =-6pt of q3, xshift=-0pt, font=\scriptsize] {$(3,4)$};
		\node[state] (q5) [below  =0.5cm of q3, xshift=-1cm] {$v_5$};
		\node[state] (lq5) [left  =-9pt of q5, xshift=-0pt, font=\scriptsize] {$(4,4)$};
		\node[state] (q7) [below  =0.5cm of q5, xshift=-1cm] {$v_7$};
		\node[state] (lq7) [left  =-6pt of q7, xshift=-0pt, font=\scriptsize] {$(4,3)$};
		\node[state] (q8) [below  =0.5cm of q5, xshift=1cm] {$v_8$};
		\node[state] (lq8) [right  =-6pt of q8, xshift=-0pt, font=\scriptsize] {$(4,3)$};
		\node[state] (q9) [below  =0.5cm of q8, xshift=-1cm] {$v_9$};
		\node[state] (lq9) [left  =-6pt of q9, xshift=-0pt, font=\scriptsize] {$(4,2)$};
		
		\node[state] (q10) [below  =0.7cm of q9, xshift=-0.5cm] {$v_{10}$};
		\node[state] (q11) [below  =0.7cm of q9, xshift=0.5cm] {$v_{11}$};
		\node[state] (q12) [below  =0.7cm of q9, xshift=-2cm] {$v_{12}$};
		\node[state] (q13) [below  =0.7cm of q9, xshift=2cm] {$v_{13}$};
		
		\path[->] 
		(q0) edge node[below=-5pt, xshift=-4pt, font=\scriptsize]{$e_0$} (q1)
		(q-1) edge node[below=-5pt, xshift=8pt, font=\scriptsize]{$e_{17}$} (q1)
		(q-2) edge node[below=-5pt, xshift=-7pt, font=\scriptsize]{$e_{18}$} (q2)
		edge node[below=-5pt, xshift=-7pt, font=\scriptsize]{$e_{20}$} (q1)
		(q-3) edge node[below=-5pt, xshift=8pt, font=\scriptsize]{$e_{19}$} (q3)
		
		(q1) edge node[below=-4pt, xshift=5pt, font=\scriptsize]{$e_1$} (q2)
		edge node[below=-4pt, xshift=-5pt, font=\scriptsize]{$e_2$} (q3)
		(q2) edge node[left=2pt, xshift=5pt, font=\scriptsize]{$s_1$} (q7)
		edge node[below=-4pt, xshift=-5pt, font=\scriptsize]{$e_4$} (q5)
		(q3) edge node[below=-4pt, xshift=5pt, font=\scriptsize]{$e_5$} (q5)
		edge node[right=2pt, xshift=-5pt, font=\scriptsize]{$s_2$} (q8)
		(q5) edge node[below=-4pt, xshift=5pt, font=\scriptsize]{$e_8$} (q7)
		edge node[below=-4pt, xshift=-5pt, font=\scriptsize]{$e_9$} (q8)
		(q7) edge node[below=-4pt, xshift=-5pt, font=\scriptsize]{$e_{11}$} (q9)
		(q8) edge node[below=-4pt, xshift=5pt, font=\scriptsize]{$e_{12}$} (q9)
		
		(q9) edge node[below=-6pt, xshift=-7pt, font=\scriptsize]{$e_{13}$} (q10)
		edge node[below=-6pt, xshift=7pt, font=\scriptsize]{$e_{14}$} (q11)
		(q7) edge node[below=-6pt, xshift=-7pt, font=\scriptsize]{$e_{15}$} (q12)
		(q8) edge node[below=-6pt, xshift=7pt, font=\scriptsize]{$e_{16}$} (q13)
		;
		\end{tikzpicture}
		\caption{}
	\end{subfigure}
    }
	\resizebox{0.3\textwidth}{!}{
	\begin{subfigure}[b]{0.4\textwidth}
		\centering
		\begin{tikzpicture}[>=stealth,thick]
			\usetikzlibrary{calc}
			\tikzstyle{state}=[]
			
			\node[state] (q0) [xshift=0cm] {$v_0$};
			\node[state] (q-1) [xshift=1cm] {$v_{-1}$};
			\node[state] (q-2) [xshift=-1.5cm] {$v_{-2}$};
			
			\node[state] (q1) [below  =0.7cm of q0, xshift=0.5cm] {$v_1$};
			\node[state] (lq1) [left  =-6pt of q1, xshift=-0pt, font=\scriptsize] {$(3,4)$};
			\node[state] (q2) [below  =0.5cm of q1, xshift=-1cm] {$v_2$};
			\node[state] (lq2) [left  =-6pt of q2, xshift=-0pt, font=\scriptsize] {$(3,4)$};
			\node[state] (q3) [below  =0.5cm of q1, xshift=1cm] {$v_3$};
			\node[state] (lq3) [right  =-6pt of q3, xshift=-0pt, font=\scriptsize] {$(3,4)$};
			\node[state] (q5) [below  =0.5cm of q3, xshift=-1cm] {$v_5$};
			\node[state] (lq5) [left  =-8pt of q5, xshift=-0pt, font=\scriptsize] {$(3,4)$};
			\node[state] (q7) [below  =0.5cm of q5, xshift=-1cm] {$v_7$};
			\node[state] (lq7) [left  =-6pt of q7, xshift=-0pt, font=\scriptsize] {$(3,3)$};
			\node[state] (q8) [below  =0.5cm of q5, xshift=1cm] {$v_8$};
			\node[state] (lq8) [right  =-6pt of q8, xshift=-0pt, font=\scriptsize] {$(3,3)$};
			\node[state] (q9) [below  =0.5cm of q8, xshift=-1cm] {$v_9$};
			\node[state] (lq9) [left  =-6pt of q9, xshift=-0pt, font=\scriptsize] {$(3,2)$};
			
			\node[state] (q10) [below  =0.7cm of q9, xshift=-0.5cm] {$v_{10}$};
			\node[state] (q11) [below  =0.7cm of q9, xshift=0.5cm] {$v_{11}$};
			\node[state] (q12) [below  =0.7cm of q9, xshift=-2cm] {$v_{12}$};
			\node[state] (q13) [below  =0.7cm of q9, xshift=2cm] {$v_{13}$};
			
			\path[->] 
			
		    (q0) edge node[below=-5pt, xshift=-4pt, font=\scriptsize]{$e_0$} (q1)
            (q-1) edge node[below=-5pt, xshift=8pt, font=\scriptsize]{$e_{17}$} (q1)
            (q-2) edge node[below=-5pt, xshift=-7pt, font=\scriptsize]{$e_{20}$} (q1)

			(q1) edge node[below=-4pt, xshift=5pt, font=\scriptsize]{$e_1$} (q2)
			     edge node[below=-4pt, xshift=-5pt, font=\scriptsize]{$e_2$} (q3)
			(q2) edge node[left=2pt, xshift=5pt, font=\scriptsize]{$s_1$} (q7)
			     edge node[below=-4pt, xshift=-5pt, font=\scriptsize]{$e_4$} (q5)
			(q3) edge node[below=-4pt, xshift=5pt, font=\scriptsize]{$e_5$} (q5)
			     edge node[right=2pt, xshift=-5pt, font=\scriptsize]{$s_2$} (q8)
			(q5) edge node[below=-4pt, xshift=5pt, font=\scriptsize]{$e_8$} (q7)
			     edge node[below=-4pt, xshift=-5pt, font=\scriptsize]{$e_9$} (q8)
			(q7) edge node[below=-4pt, xshift=-5pt, font=\scriptsize]{$e_{11}$} (q9)
			(q8) edge node[below=-4pt, xshift=5pt, font=\scriptsize]{$e_{12}$} (q9)

		    (q9) edge node[below=-6pt, xshift=-7pt, font=\scriptsize]{$e_{13}$} (q10)
                 edge node[below=-6pt, xshift=7pt, font=\scriptsize]{$e_{14}$} (q11)
            (q7) edge node[below=-6pt, xshift=-7pt, font=\scriptsize]{$e_{15}$} (q12)
            (q8) edge node[below=-6pt, xshift=7pt, font=\scriptsize]{$e_{16}$} (q13)
			
			;
		\end{tikzpicture}
		\caption{}
    \end{subfigure}

	}
	\resizebox{0.3\textwidth}{!}{
		\begin{subfigure}[b]{0.4\textwidth}
			\centering
			\begin{tikzpicture}[>=stealth,thick]
			\usetikzlibrary{calc}
			\tikzstyle{state}=[]
			
			\node[state] (q-2) [xshift=-1.5cm] {$v_{-2}$};
			\node[state] (q-3) [xshift=2.5cm] {$v_{-3}$};
			
			\node[state] (q1) [below  =0.7cm of q0, xshift=0.5cm] {};
			\node[state] (q2) [below  =0.5cm of q1, xshift=-1cm] {$v_2$};
			\node[state] (lq2) [left  =-6pt of q2, xshift=-0pt, font=\scriptsize] {$(1,4)$};
			\node[state] (q3) [below  =0.5cm of q1, xshift=1cm] {$v_3$};
			\node[state] (lq3) [right  =-6pt of q3, xshift=-0pt, font=\scriptsize] {$(1,4)$};
			\node[state] (q5) [below  =0.5cm of q3, xshift=-1cm] {$v_5$};
			\node[state] (lq5) [left  =-8pt of q5, xshift=-0pt, font=\scriptsize] {$(2,4)$};
			\node[state] (q7) [below  =0.5cm of q5, xshift=-1cm] {$v_7$};
			\node[state] (lq7) [left  =-6pt of q7, xshift=-0pt, font=\scriptsize] {$(2,3)$};
			\node[state] (q8) [below  =0.5cm of q5, xshift=1cm] {$v_8$};
			\node[state] (lq8) [right  =-6pt of q8, xshift=-0pt, font=\scriptsize] {$(2,3)$};
			\node[state] (q9) [below  =0.5cm of q8, xshift=-1cm] {$v_9$};
			\node[state] (lq9) [left  =-6pt of q9, xshift=-0pt, font=\scriptsize] {$(2,2)$};
			
			\node[state] (q10) [below  =0.7cm of q9, xshift=-0.5cm] {$v_{10}$};
			\node[state] (q11) [below  =0.7cm of q9, xshift=0.5cm] {$v_{11}$};
			\node[state] (q12) [below  =0.7cm of q9, xshift=-2cm] {$v_{12}$};
			\node[state] (q13) [below  =0.7cm of q9, xshift=2cm] {$v_{13}$};
			
			\path[->] 
			(q-2) edge (q2)
			(q-3) edge (q3)
			
			(q2) edge node[left=2pt, xshift=5pt, font=\scriptsize]{$s_1$} (q7)
			     edge node[below=-4pt, xshift=-5pt, font=\scriptsize]{$e_4$} (q5)
			(q3) edge node[below=-4pt, xshift=5pt, font=\scriptsize]{$e_5$} (q5)
			     edge node[right=2pt, xshift=-5pt, font=\scriptsize]{$s_2$} (q8)
			(q5) edge node[below=-4pt, xshift=5pt, font=\scriptsize]{$e_8$} (q7)
			     edge node[below=-4pt, xshift=-5pt, font=\scriptsize]{$e_9$} (q8)
			(q7) edge node[below=-4pt, xshift=-5pt, font=\scriptsize]{$e_{11}$} (q9)
			(q8) edge node[below=-4pt, xshift=5pt, font=\scriptsize]{$e_{12}$} (q9)
			
			(q9) edge (q10)
			     edge (q11)
			(q7) edge (q12)
			(q8) edge (q13)
			;
			\end{tikzpicture}
			\caption{}
		\end{subfigure}
	}
    \resizebox{0.3\textwidth}{!}{
    	\begin{subfigure}[b]{0.4\textwidth}
    		\centering
    		\begin{tikzpicture}[>=stealth,thick]
    		\usetikzlibrary{calc}
    		\tikzstyle{state}=[]
    		
    		\node[state] (q0) [xshift=0cm] {$v_0$};
    		\node[state] (q-1) [xshift=1cm] {$v_{-1}$};
    		\node[state] (q-2) [xshift=-1.5cm] {$v_{-2}$};
    		
    		\node[state] (q1) [below  =0.7cm of q0, xshift=0.5cm] {$v_1$};
    		\node[state] (lq1) [left  =-6pt of q1, xshift=-0pt, font=\scriptsize] {$(3,4)$};
    		\node[state] (q2) [below  =0.5cm of q1, xshift=-1cm] {};
    		\node[state] (q3) [below  =0.5cm of q1, xshift=1cm] {};
    		\node[state] (q4) [below  =0.5cm of q2, xshift=-1cm] {};
    		\node[state] (q5-1) [below  =0.5cm of q2, xshift=0.5cm] {};
    		\node[state] (q5-2) [below  =0.5cm of q3, xshift=-0.5cm] {};
    		\node[state] (q6) [below  =0.5cm of q3, xshift=1cm] {};
    		\node[state] (q7) [below  =0.5cm of q5-1, xshift=-0.5cm] {$v_7$};
    		\node[state] (lq7) [left  =-6pt of q7, xshift=-0pt, font=\scriptsize] {$(3,3)$};
    		\node[state] (q8) [below  =0.5cm of q5-2, xshift=0.5cm] {$v_8$};
    		\node[state] (lq8) [right  =-6pt of q8, xshift=-0pt, font=\scriptsize] {$(3,3)$};
    		\node[state] (q9) [below  =0.5cm of q8, xshift=-1cm] {$v_9$};
    		\node[state] (lq9) [left  =-6pt of q9, xshift=-0pt, font=\scriptsize] {$(3,2)$};
    		
    		\node[state] (q10) [below  =0.7cm of q9, xshift=-0.5cm] {$v_{10}$};
    		\node[state] (q11) [below  =0.7cm of q9, xshift=0.5cm] {$v_{11}$};
    		\node[state] (q12) [below  =0.7cm of q9, xshift=-2cm] {$v_{12}$};
    		\node[state] (q13) [below  =0.7cm of q9, xshift=2cm] {$v_{13}$};
    		
    		\path[->] 
    		(q0) edge node[below=-4pt, xshift=-5pt, font=\scriptsize]{$e_0$} (q1)
    		(q-1) edge node[below=-4pt, xshift=7pt, font=\scriptsize]{$e_{15}$}(q1)
    		(q-2) edge (q1)

    		(q1) edge node[below=-4pt, xshift=-7pt, font=\scriptsize]{$s_8$} (q7)
    		     edge node[below=-4pt, xshift=7pt, font=\scriptsize]{$s_9$} (q8)

    		(q7) edge node[below=-4pt, xshift=-5pt, font=\scriptsize]{$e_{11}$} (q9)
    		(q8) edge node[below=-4pt, xshift=5pt, font=\scriptsize]{$e_{12}$} (q9)
    		
    		(q9) edge (q10)
    		edge (q11)
    		(q7) edge (q12)
    		(q8) edge (q13)
    		;
    		\end{tikzpicture}
    		\caption{}
    	\end{subfigure}
    	
    }
    \resizebox{0.3\textwidth}{!}{
    	\begin{subfigure}[b]{0.4\textwidth}
    		\centering
    		\begin{tikzpicture}[>=stealth,thick]
    		\usetikzlibrary{calc}
    		\tikzstyle{state}=[]
    		
    		\node[state] (q-2) [xshift=-1.5cm] {$v_{-2}$};
    		\node[state] (q-3) [xshift=2.5cm] {$v_{-3}$};
    		
    		\node[state] (q1) [below  =0.7cm of q0, xshift=0.5cm] {};
    		\node[state] (q2) [below  =0.5cm of q1, xshift=-1cm] {$v_2$};
    		\node[state] (lq2) [left  =-6pt of q2, xshift=-0pt, font=\scriptsize] {$(1,2)$};
    		\node[state] (q3) [below  =0.5cm of q1, xshift=1cm] {$v_3$};
    		\node[state] (lq3) [right  =-6pt of q3, xshift=-0pt, font=\scriptsize] {$(1,2)$};
    		\node[state] (q5) [below  =0.5cm of q3, xshift=-1cm] {$v_5$};
    		\node[state] (lq5) [left  =-8pt of q5, xshift=-0pt, font=\scriptsize] {$(2,2)$};
    		\node[state] (q7) [below  =0.5cm of q5, xshift=-1cm] {$v_7$};
    		\node[state] (lq7) [left  =-6pt of q7, xshift=-0pt, font=\scriptsize] {$(2,1)$};
    		\node[state] (q8) [below  =0.5cm of q5, xshift=1cm] {$v_8$};
    		\node[state] (lq8) [right  =-6pt of q8, xshift=-0pt, font=\scriptsize] {$(2,1)$};
    		\node[state] (q12) [below  =0.7cm of q9, xshift=-2cm] {$v_{12}$};
    		\node[state] (q13) [below  =0.7cm of q9, xshift=2cm] {$v_{13}$};
    		
    		\path[->] 

    		(q-2) edge (q2)
    		(q-3) edge (q3)

    		(q2) edge node[left=2pt, xshift=5pt, font=\scriptsize]{$s_1$} (q7)
    		     edge node[below=-4pt, xshift=-5pt, font=\scriptsize]{$e_4$} (q5)
    		(q3) edge node[below=-4pt, xshift=5pt, font=\scriptsize]{$e_5$} (q5)
    		     edge node[right=2pt, xshift=-5pt, font=\scriptsize]{$s_2$} (q8)
    		(q5) edge node[below=-4pt, xshift=5pt, font=\scriptsize]{$e_8$} (q7)
    		     edge node[below=-4pt, xshift=-5pt, font=\scriptsize]{$e_9$} (q8)

    		(q7) edge (q12)
    		(q8) edge (q13)
    		;
    		\end{tikzpicture}
    		\caption{}
    	\end{subfigure}

    }
    \resizebox{0.3\textwidth}{!}{
    	\begin{subfigure}[b]{0.4\textwidth}
    		\centering
    		\begin{tikzpicture}[>=stealth,thick]
    		\usetikzlibrary{calc}
    		\tikzstyle{state}=[]
    		
    		\node[state] (q-2) [xshift=-1.5cm] {$v_{-2}$};
    		\node[state] (q-3) [xshift=2.5cm] {$v_{-3}$};
    		
    		\node[state] (q1) [below  =0.7cm of q0, xshift=0.5cm] {};
    		\node[state] (q2) [below  =0.5cm of q1, xshift=-1cm] {$v_2$};
    		\node[state] (lq2) [left  =-6pt of q2, xshift=-0pt, font=\scriptsize] {$(1,2)$};
    		\node[state] (q3) [below  =0.5cm of q1, xshift=1cm] {$v_3$};
    		\node[state] (lq3) [right  =-6pt of q3, xshift=-0pt, font=\scriptsize] {$(1,2)$};
    		\node[state] (q5) [below  =0.5cm of q3, xshift=-1cm] {$v_5$};
    		\node[state] (lq5) [left  =-8pt of q5, xshift=-0pt, font=\scriptsize] {$(2,2)$};
    		\node[state] (q7) [below  =0.5cm of q5, xshift=-1cm] {$v_7$};
    		\node[state] (lq7) [left  =-6pt of q7, xshift=-0pt, font=\scriptsize] {$(2,2)$};
    		\node[state] (q8) [below  =0.5cm of q5, xshift=1cm] {$v_8$};
    		\node[state] (lq8) [right  =-6pt of q8, xshift=-0pt, font=\scriptsize] {$(2,2)$};
    		\node[state] (q9) [below  =0.5cm of q8, xshift=-1cm] {$v_9$};
    		\node[state] (lq9) [left  =-6pt of q9, xshift=-0pt, font=\scriptsize] {$(2,2)$};
    		
    		\node[state] (q10) [below  =0.7cm of q9, xshift=-0.5cm] {$v_{10}$};
    		\node[state] (q11) [below  =0.7cm of q9, xshift=0.5cm] {$v_{11}$};
    		
    		\path[->] 
    		(q-2) edge (q2)
    		(q-3) edge (q3)
    		
    		(q2) edge node[left=2pt, xshift=5pt, font=\scriptsize]{$s_1$} (q7)
    		     edge node[below=-4pt, xshift=-5pt, font=\scriptsize]{$e_4$} (q5)
    		(q3) edge node[below=-4pt, xshift=5pt, font=\scriptsize]{$e_5$} (q5)
    		     edge node[right=2pt, xshift=-5pt, font=\scriptsize]{$s_2$} (q8)
    		(q5) edge node[below=-4pt, xshift=5pt, font=\scriptsize]{$e_8$} (q7)
    		edge node[below=-4pt, xshift=-5pt, font=\scriptsize]{$e_9$} (q8)
    		(q7) edge node[below=-4pt, xshift=-5pt, font=\scriptsize]{$e_{11}$} (q9)
    		(q8) edge node[below=-4pt, xshift=5pt, font=\scriptsize]{$e_{12}$} (q9)
    		
    		(q9) edge (q10)
    		     edge (q11)
    		;
    		\end{tikzpicture}
    		\caption{}
    	\end{subfigure}

    }
    \resizebox{0.3\textwidth}{!}{
    	\begin{subfigure}[b]{0.4\textwidth}
    		\centering
    		\begin{tikzpicture}[>=stealth,thick]
    		\usetikzlibrary{calc}
    		\tikzstyle{state}=[]
    		
    		\node[state] (q-2) [xshift=-1.5cm] {$v_{-2}$};
    		\node[state] (q-3) [xshift=2.5cm] {$v_{-3}$};
    		
    		\node[state] (q1) [below  =0.7cm of q0, xshift=0.5cm] {};
    		\node[state] (q9) [below  =0.5cm of q8, xshift=-1cm] {$v_9$};
    		\node[state] (lq9) [left  =-6pt of q9, xshift=-0pt, font=\scriptsize] {$(2,2)$};
    		
    		\node[state] (q10) [below  =0.7cm of q9, xshift=-0.5cm] {$v_{10}$};
    		\node[state] (q11) [below  =0.7cm of q9, xshift=0.5cm] {$v_{11}$};

    		\path[->] 
    		(q-2) edge (q9)
    		(q-3) edge (q9)

    		(q9) edge (q10)
    		     edge (q11)
    		;
    		\end{tikzpicture}
    		\caption{}
    	\end{subfigure}
    	
    }

	\caption{Multi-Roots or Multi-Terminals Factorization}
	\label{fig:multi-root-terminal-factorization}
\end{figure}
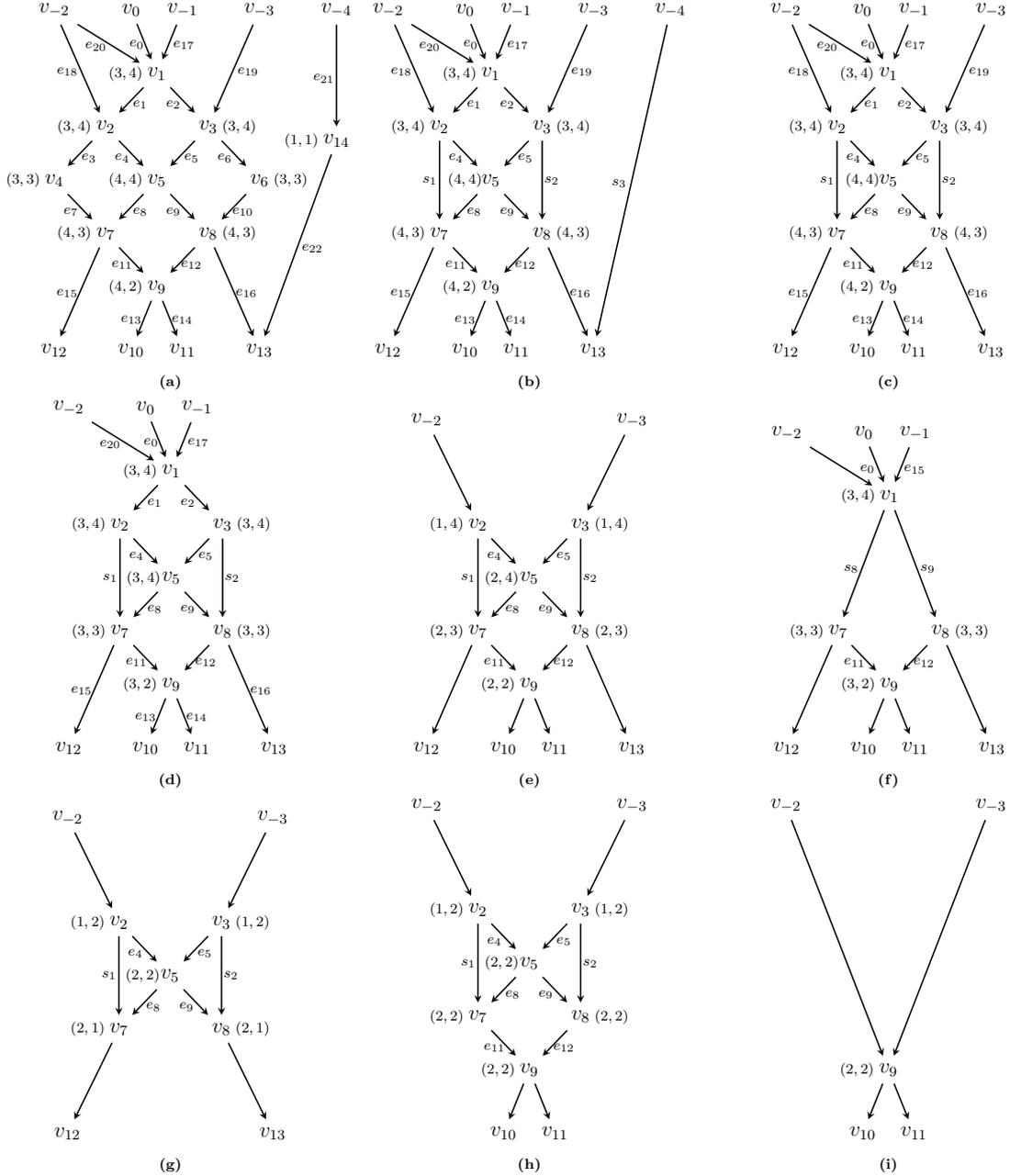

\noindent 
Our strategy of finding common sub-chains and sub-blocks: \\ 
1. If there are simple chains or simple blocks, replace them with reference edges. 
e.g., after replacing all the simple chains in \cref*{fig:multi-root-terminal-factorization} (a), 
it becomes \cref*{fig:multi-root-terminal-factorization} (b). 
\\
2. Find a vertex $v_i$ that has the largest r-degree, 
$\forall v_k\in V^{L(v_i)}\cap v_k\ne v_i, D^{-}_r(v_k)<D^{-}_r(v_i)$, 
in the depth level closest to the first level; 
Find a vertex $v_j$ in $P_{v_i\to X}$ 
that has the largest t-degree, 
$\forall v_k\in V^{L(v_j)}\cap v_k\ne v_j, D^{+}_t(v_k)<D^{+}_t(v_j)$,
in the depth level closest to the last level. 
e.g., in \cref*{fig:multi-root-terminal-factorization} (b), we find that $v_i=v_j=v_5$. 
\\
3. If $Y_i=Y$ and $X_j=X$, go to step 4. 
If $Y_i \ne Y$ or $X_j \ne X$, then 
split all the edges and intermediate vertices in $P_{Y_i\to X_j}$, i.e. $P_{Y_i^{+}\to X_j^{-}}$, 
into a new page from the current page, go to step 6. 
e.g., we split \cref*{fig:multi-root-terminal-factorization} (b) into two pages, 
one of which is \cref*{fig:multi-root-terminal-factorization} (c), 
and the edge between $v_{-4}$ and $v_{13}$ is on the other page. 
Note that we do not split roots and terminals when splitting pages. 
All the roots are still in the same line of the first depth level, 
and all the terminals are still in the same line of the last depth level. 
Different pages share the same lines of roots and terminals. 
\\ 
4. If 
$L(v_i)=L(v_j)$ or 
$(|P_{v_i\to v_j}|=1)\cap(|p_{v_i\to v_j}|=1)$, i.e., $||P_{v_i\to v_j}||=1$, 
then separate root-terminal pairs as follows: 
\\ 
If there are cross-level outgoing edges from root vertices 
or cross-level incoming edges to terminal vertices, then, 
let $E'$ be the set of all outgoing edges that cross most levels from roots, 
separate $P_{E'\to X^{-}}$ from the current page, 
or, 
let $E'$ be the set of all incoming edges that cross most levels to terminals, 
separate $P_{Y^{+}\to E'}$ from the current page. 
e.g., we separate \cref*{fig:multi-root-terminal-graph} (c) into 
\cref*{fig:multi-root-terminal-factorization} (d) 
and \cref*{fig:multi-root-terminal-factorization} (e), 
separate \cref*{fig:multi-root-terminal-factorization} (e) into 
\cref*{fig:multi-root-terminal-factorization} (g) 
and \cref*{fig:multi-root-terminal-factorization} (h). 
Otherwise, 
if $L(v_i)-L(Y)\ge L(X)-L(v_j)$, select a set of roots $Y'$, separate $P_{Y'^{+}\to X^{-}}$; 
if $L(v_i)-L(Y)\le L(X)-L(v_j)$, select a set of terminals $X'$, separate $P_{Y^{+} \to X'^{-}}$; 
Go to step 6. 
\\
5. If $L(v_i)<L(v_j)$ and $||P_{v_i\to v_j}||>1$: \\ 
split $P_{(V^{L(v_i)}_j)^{+}\to (V^{L(v_j)}_i)^{-}}$ into a new page and go to step 6
if corresponding $v_i$ and $v_j$ can be found in it, 
otherwise factorize the vertices and edges in $P_{V^{L(v_i)}_j\to V^{L(v_j)}_i}$ 
locally 
and replace simple chains or simple blocks with reference edges  
until $L(v_i)=L(v_j)$ or $||P_{v_i\to v_j}||=1$. 
e.g., we factorize the vertices and edges between $V^1$ and $V^3$ 
in \cref*{fig:multi-root-terminal-factorization} (d) 
and replace simple chains and simple blocks with reference edges, 
the result page is \cref*{fig:multi-root-terminal-factorization} (f). 
\\
6. Apply the above process to the separated pages 
until there is only one root-terminal pair or
no such $v_i$ and $v_j$ can be found in each page. 
Merge the same root-terminal pairs in different pages. 
\\ \\

\begin{figure}[H] 
	\centering
	\resizebox{.3\textwidth}{!}{
		\begin{subfigure}[b]{0.5\textwidth}
			\centering
			\begin{tikzpicture}[>=stealth,thick]
			\usetikzlibrary{calc}
			\tikzstyle{state}=[]
			
			\node[state] (y2) [xshift=0cm] {$y_2$};
			\node[state] (y1) [xshift=7cm] {$y_1$};
			
			\node[state] (e18) [below  =0.7cm of y2, xshift=-0cm] {$e_{18}$};
			\node[state] (e19) [below  =0.7cm of y1, xshift=0cm] {$e_{19}$};
			
			\node[state] (e3e7) [below  =1.2cm of e1, xshift=-1cm] {$e_{3}e_{7}$};
			\node[state] (e4) [below  =0.7cm of e1, xshift=-0cm] {$e_{4}$};
			\node[state] (e5) [below  =0.7cm of e2, xshift=0cm] {$e_{5}$};
			\node[state] (e6e10) [below  =1.2cm of e2, xshift=1cm] {$e_{6}e_{10}$};
			
			\node[state] (e8) [below  =0.7cm of e4, xshift=-0cm] {$e_{8}$};
			\node[state] (e9) [below  =0.7cm of e5, xshift=0cm] {$e_{9}$};
			
			\node[state] (e11) [below  =0.7cm of e8, xshift=0.0cm] {$e_{11}$};
			\node[state] (e12) [below  =0.7cm of e9, xshift=-0.0cm] {$e_{12}$};
			
			\node[state] (e15) [below  =0.7cm of e11, xshift=-3cm] {$e_{15}$};
			\node[state] (e13) [below  =0.7cm of e11, xshift=-0cm] {$e_{13}$};
			\node[state] (e14) [below  =0.7cm of e12, xshift=0cm] {$e_{14}$};
			\node[state] (e16) [below  =0.7cm of e12, xshift=3cm] {$e_{16}$};
			
			\node[state] (x3) [below  =0.7cm of e15, xshift=-0cm] {$x_{3}$};
			\node[state] (x1) [below  =0.7cm of e13, xshift=-0cm] {$x_{1}$};
			\node[state] (x2) [below  =0.7cm of e14, xshift=0cm] {$x_{2}$};
			\node[state] (x4) [below  =0.7cm of e16, xshift=0cm] {$x_{4}$};
			
			\path[->] 
			(y2) edge (e18)
			(y1) edge (e19)
			
			(e18) edge (e3e7)
			edge (e4)
			(e19) edge (e5)
			      edge (e6e10)
			
			(e4) edge (e8)
			edge (e9)
			(e5) edge (e8)
			edge (e9)
			
			(e3e7) edge (e11)
			       edge (e15)
			(e8) edge (e11)
			     edge (e15)
			(e9) edge (e12)
			     edge (e16)
			(e6e10) edge (e12)
			      edge (e16)
			
			(e11) edge (e13)
			edge (e14)
			(e12) edge (e13)
			edge (e14)
			
			(e15) edge (x3)
			(e13) edge (x1)
			(e14) edge (x2)
			(e16) edge (x4)
			
			;
			\end{tikzpicture}
			\caption{}
		\end{subfigure}
	}
	\caption{}\label{fig:line-graph-eg}
\end{figure}

\noindent

\begin{flalign}
&s_1=e_3e_7 \label{eq:s1} \\
&s_2=e_6e_{10} \label{eq:s2}\\
&s_4=(e_8e_{11}+e_9e_{12}) \label{eq:s4} \\
&s_5=e_{18}(s_1e_{11}+e_4s_4)=e_{18}(e_3e_7e_{11}+e_4(e_8e_{11}+e_9e_{12})) \label{eq:s5} \\
&s_6=e_{19}(s_2e_{12}+e_5s_4)=e_{19}(e_6e_{10}e_{12}+e_5(e_8e_{11}+e_9e_{12})) \label{eq:s6} \\
&\frac{\partial v_{-2}}{\partial v_{12}}=e_{18}(s_1+e_4e_8)e_{15}
=e_{18}(e_3e_7+e_4e_8)e_{15} \label{eq:v-2_v12}\\
&\frac{\partial v_{-2}}{\partial v_{13}}=e_{18}e_4e_9e_{16} \label{eq:v-2_v13}\\
&\frac{\partial v_{-2}}{\partial v_{10}}=s_5e_{13}
=e_{18}(e_3e_7e_{11}+e_4(e_8e_{11}+e_9e_{12}))e_{13} \label{eq:v-2_v10}\\
&\frac{\partial v_{-2}}{\partial v_{11}}=s_5e_{14}
=e_{18}(e_3e_7e_{11}+e_4(e_8e_{11}+e_9e_{12}))e_{14} \label{eq:v-2_v11}\\
&\frac{\partial v_{-3}}{\partial v_{13}}=e_{19}(e_5e_9+s_2)e_{16}
=e_{19}(e_5e_9+e_6e_{10})e_{16} \label{eq:v-3_v13}\\
&\frac{\partial v_{-3}}{\partial v_{12}}=e_{19}e_5e_8e_{15} \label{eq:v-3_v12}\\
&\frac{\partial v_{-3}}{\partial v_{10}}=s_6e_{13}
=e_{19}(e_6e_{10}e_{12}+e_5(e_8e_{11}+e_9e_{12}))e_{13} \label{eq:v-3_v10}\\
&\frac{\partial v_{-3}}{\partial v_{11}}=s_6e_{14}
=e_{19}(e_6e_{10}e_{12}+e_5(e_8e_{11}+e_9e_{12}))e_{14} \label{eq:v-3_v11}
\end{flalign}

We notice that every edge between two intermediate vertices in line-graph 
is a multiplication relation, 
however in algebraic expressions, 
there are two types of multiplication relation: direct and indirect. 
e.g., in $e_{18}(e_3e_7+e_4e_8)e_{15}$, 
there is a direct multiplication relation between $e_4$ and $e_8$, 
whereas in $e_4(e_8e_{11}+e_9e_{12})$, 
the multiplication relation between $e_4$ and $e_8$ is indirect, 
there is a parenthesis between $e_4$ and $e_8$.

Assume we have found the optimal calculation order of a differentiation graph $G$, 
which is represented by a set of algebraic expressions connected by reference variables. 
In its corresponding line-graph $\tilde{G}$, 
there may be different types of multiplication relations 
in different root-terminal pairs for the same edge, 
elimination of which may break parentheses in the optimal algebraic expressions 
of some root-terminal pairs, 
i.e. break the optimal calculation order. 
e.g., \cref*{fig:line-graph-eg} is 
the line-graph of \cref*{fig:multi-root-terminal-factorization} (e), 
we assume \cref*{eq:s1} - \cref*{eq:v-3_v11} is an optimal calculation order. 
If we eliminate edge $\langle e_{18},e_4\rangle$ directly from \cref*{fig:line-graph-eg} for \cref*{eq:v-2_v13}, 
it will break the calculation order of \cref*{eq:v-2_v12}, \cref*{eq:v-2_v10} 
and \cref*{eq:v-2_v11}. 

We can fix it by explicitly splitting edges with conflicting multiplication relations, 
separating the direct multiplication relation from the indirect, 
and only eliminating those edges that represent direct multiplication relations. 
However, this method not only copies a lot of vertices and edges, 
but also requires more decision-making about how to split and which ones to eliminate. 
Furthermore, the splitting of an edge sometimes depends on the splitting of another adjacent  edge, 
e.g., in \cref*{fig:line-graph-eg}, 
the splitting of $\langle e_{18},e_4\rangle$ 
depends on the splitting of $\langle e_{4},e_9\rangle$.

Another solution is to search for a face elimination sequence 
without explicit splitting. 
Assume we have already known a set of optimal algebraic expressions, 
we can eliminate edges in line-graph in the same order as 
the calculation order of the corresponding set of optimal algebraic expressions, 
avoiding breaking parentheses. 
Unlike parentheses elimination in an algebraic expression, 
the face elimination of a sub-expression block in line-graph 
may not only depend on the sub-expression blocks nested in it 
but also the sub-expression blocks outside it, 
because 
one edge that denotes direct multiplication relation 
in a sub-expression block of an expression 
may denote indirect multiplication relation in a sub-expression block of another expression. 
e.g., the face elimination of $(e_3e_7+e_4e_8)$ in \cref*{eq:v-2_v12} 
depends on the face elimination of $e_8e_{11}$ in \cref*{eq:s5}. 
Therefore when eliminating multiplication relations inside sub-expression blocks, 
we need to eliminate those direct multiplication relations 
that have no indirect multiplication relations in other sub-expression blocks first. 
e.g., elimination of $\langle e_{8},e_{11}\rangle$ and $\langle e_{9},e_{12}\rangle$ in \cref*{fig:line-graph-eg} 
will not break the optimal calculation order in the algebraic expression set
\cref*{eq:s1} - \cref*{eq:v-3_v11}.

\begin{lemma}\label{lemma:1}
Assume that $G(V,E)$ is a differentiation graph, 
$A$ is one of its corresponding optimal algebraic expression sets, 
$s_1$, $s_2$ are edge or reference edges in $G$ and $s_1\succ s_2$. 
If $s_1$ and $s_2$ have both direct and indirect multiplication relations in $A$, 
then there must be an edge or reference edge $s_3$ 
in the indirect multiplication relation of $s_1$ and $s_2$, 
such that $s_1(s_2s_3+\Delta s)$ or $(\Delta s+s_3s_1)s_2$, 
where $\Delta s$ is the rest sub-expressions in the block 
and $\Delta s\ne0\cap s_3\notin\Delta s$. 
\end{lemma}
\begin{proof}
There are 4 possible forms of indirect multiplication relation for $s_1$ and $s_2$:
$s_1(s_2+\Delta s)$, $s_1(s_2s_3+\Delta s)$, $(\Delta s+s_1)s_2$ and $(\Delta s+s_3s_1)s_2$. 
Assume that $a_1\in A\ni a_2$, 
and $\exists s_1(s_2+\Delta s)\in a_1\cap s_i(s_1s_2+\Delta s_a)s_j\in a_2$, 
$s_i$ is the predecessor of $(s_1s_2+\Delta s_a)$, 
$s_j$ is the successor of $(s_1s_2+\Delta s_a)$, 
$\Delta s_a$ is the sibling of $s_1s_2$, 
$s_i=1$ if it has no predecessor, 
$s_j=1$ if it has no successor, 
$\Delta s_a=0$ if it has no sibling. 
From the sub-expressions $s_1(s_2+\Delta s)$ we know that $s_2$ and $\Delta s$ 
have the same source vertex and destination vertex, 
i.e., $s_2^{-}=\Delta s^{-}\cap s_2^{+}=\Delta s^{+}$, 
therefore there must be a sibling sub-expression of $s_i(s_1s_2+\Delta s_a)s_j$, 
such that $
s_i(s_1\Delta s+\Delta s_b)s_j+s_i(s_1s_2+\Delta s_a)s_j\in a_2$, 
where $\Delta s_b=0$ if $s_1\Delta s$ has no sibling. 
Obviously, $s_i(s_1\Delta s+\Delta s_b)s_j+s_i(s_1s_2+\Delta s_a)s_j=
s_i(s_1(\Delta s+s_2)+\Delta s_b)s_j$, 
the number of calculation can be reduced, which means $a_2$ and $A$ are not optimal. 
Therefore, $s_1(s_2+\Delta s)$ is impossible. 
Similarly, we can prove $(\Delta s+s_1)s_2$ is also impossible. 
\end{proof} 

Notice that we can omit $\Delta s$ in the indirect multiplication relations above, 
simply use $s_1|s_2s_3$ and $s_3s_1|s_2$ to denote 
$s_1(s_2s_3+\Delta s)$ and $(\Delta s+s_3s_1)s_2$ respectively. 
 
\begin{lemma}\label{lemma:2} 
Assume that $G(V,E)$ is a differentiation graph, 
$A$ is one of its corresponding optimal algebraic expression sets, 
$s_1$, $s_2$ are edge or reference edges in $G$ and $s_1\succ s_2$. 
If $s_1$ and $s_2$ have both direct and indirect multiplication relations in $A$, 
then in order to avoid breaking parentheses between $s_1$ and $s_2$ in $A$, 
the face elimination of $\langle s_1,s_2\rangle$ in line-graph 
depends on the face eliminations of all $\langle s_2,s_i\rangle$ that have $s_1|s_2s_i$ in $A$. 
\end{lemma}
\begin{proof}
For all the edges $\langle s_2,s_i\rangle$ in line-graph that have $s_1|s_2s_i$ in $A$, 
the elimination of $\langle s_2,s_i\rangle$ 
will implicitly split $s_2$ into another vertex 
and create a new edge connecting $s_1$ to it. 
\end{proof} 

Notice that It is possible for direct and indirect multiplication relations 
to form circular structure in an algebraic expression set,  
because there are two directions in $s_1|s_2s_3$ and $s_3s_1|s_2$. 
e.g., $s_1|s_2s_3$, $s_2|s_3s_4$, 
$s_5s_3|s_4$, $s_6s_5|s_3$, $s_7s_6|s_5$, $s_8s_7|s_6$, 
$s_8|s_7s_1$, $s_7|s_1s_2$. 
It seems impossible for us to eliminate an edge in the circular dependencies 
without breaking parentheses. 

\newpage
\begin{proposition}\label{proposition:1}
Given a differentiation graph, 
face elimination can achieve optimal Jacobian accumulation 
that requires minimum number of calculations 
if there exists at least one corresponding optimal algebraic expression set 
that doesn't contain circular elimination dependencies.  
\end{proposition}
\begin{proof} 
Assume that $G(V,E)$ is a differentiation graph, 
$A$ is one of its corresponding optimal algebraic expression sets 
that doesn't contain circular elimination dependencies, 
$s_1$, $s_2$ are edge or reference edges in $G$ and $s_1\succ s_2$. 
If $s_1$ and $s_2$ have both direct and indirect multiplication relations in $A$, 
then for all the edges $\langle s_2,s_i\rangle$ in line-graph that have $s_1|s_2s_i$ in $A$, 
if $\langle s_2,s_i\rangle$ has no indirect multiplication relation in $A$, 
we can eliminate it first, 
if all $\langle s_2,s_i\rangle$ that have $s_1|s_2s_i$ in $A$ are eliminated, 
we can safely eliminate the edge $\langle s_1,s_2\rangle$ in line-graph 
without worrying about accidentally breaking parentheses in any $s_1|s_2s_i$. 
If $\langle s_2,s_i\rangle$ also has indirect multiplication relations $s_2|s_is_j$ in $A$, 
we can apply this process to $s_2s_i$ and $s_2|s_is_j$ recursively. 
\end{proof}

We don't yet know if it is impossible to form circular dependencies 
in the optimal algebraic expression sets of a differentiation graph. 
If all the optimal algebraic expression sets of a differentiation graph 
contain circular dependencies, 
we may not be able to achieve optimal by performing face elimination. 
\\

\section{Conclusion}

There is close correspondence between differentiation graph and algebraic expression, 
simple differentiation graph and algebraic expression are interconvertible 
provided that the multiplication relations in the algebraic expressions 
converted from differentiation graphs are noncommutative. 
Differentiation graph is more flexible, 
it can merge the same intermediate variable or subexpression 
that cannot be factored out in algebraic expression. 
There is no direct representation or corresponding counterpart of complex block 
in algebraic expression. 
We can factorize overly merged complex blocks or chains, 
and transform them into simple differentiation graphs only with simple chains and simple blocks. 
For the simple blocks and simple chains, 
elimination of paths is equivalent to elimination of corresponding edges and vertices 
as long as it follows the optimal calculation order. 
Local Jacobian multiplication 
can automatically factorize complex blocks
but it is up to us to find a better accumulation order. 
The purpose of multiple blocks or chains being intertwined together is 
to merge the same common sub-chains or sub-blocks, 
if we know the optimal common sub-chain or sub-block combination set, we can compute 
them separately. 
In line-graph there may be different types of multiplication relations for the same edge, 
elimination of which may break the optimal calculation order. 
Unlike parentheses elimination in an algebraic expression, 
the face elimination of a sub-expression block in line-graph 
may not only depend on the sub-expression blocks nested in it 
but also the sub-expression blocks outside it. 
It is possible for direct and indirect multiplication relations to form circular dependencies 
in an algebraic expression set. 
If all the optimal algebraic expression sets of a differentiation 
graph contain circular dependencies, 
we may not be able to achieve optimal by performing face elimination.

\newpage
\nocite{*} 

\bibliography{chain-rule.bib}
\bibliographystyle{plain} 
\end{document}